\documentclass[12pt]{amsart}

\usepackage{cite}
\usepackage{xypic}
\usepackage{graphicx}
\usepackage{amsfonts}
\usepackage{mathrsfs}
\usepackage{amssymb}
\usepackage{amsxtra}
\def\G{\Gamma}
\def\del{\partial}
\def\SS{\mathbb S}
\def\Z{\mathbb Z}
\def\T{\mathbb T}
\def\R{\mathbb R}
\def\C{\mathbb C}
\def\TTheta{{\mathbb T}_{\Theta}}

\def\B{\mathscr B}
\def\BG{\bar \Gamma}

\def\nn{\nonumber}
\def\t{\tau}
\def\A{\mathscr{A}}
\def\H{\mathscr{H}}

\def\stab{St}
\def\wt{wt}
\def\group{Sym(\BG)}

\def\hrep{\rho}
\def\mrep{Mat}
\def\urep{\nu}
\def\rep{(\H,U)}
\def\symrep{\alpha}
\def\bt{\boldsymbol{\tau}}
\def\coh{\alpha}

\def\Gpd{\mathcal G}
\def\P{\mathcal P}
\def\Aut{\mathcal L}

\def\Atot{\tilde A}
\def\Btot{\tilde B}

\newtheorem{thm}{Theorem}[section]
\newtheorem{lem}[thm]{Lemma}
\newtheorem{prop}[thm]{Proposition}
\newtheorem{cor}[thm]{Corollary}

\theoremstyle{definition}

\newtheorem{rmk}[thm]{Remark}

\newtheorem{ex}[thm]{Example}

\newcommand{\cev}[1]{\stackrel{\leftarrow}{#1}}
\renewcommand{\vec}[1]{\stackrel{\rightarrow}{#1}}
\accentedsymbol{\dbarG}{\Bar{\Bar{\Gamma}}}

\begin{document}

\title[Re-gauging groupoid, symmetries and degeneracies]{
 Re-gauging groupoid, symmetries and degeneracies for Graph Hamiltonians
and applications
to the Gyroid wire network}

\author
[Ralph M.\ Kaufmann]{Ralph M.\ Kaufmann}
\email{rkaufman@math.purdue.edu}

\address{Department of Mathematics, Purdue University,
 West Lafayette, IN 47907 and Institute for Advanced Study, Princeton NJ, 08540}

\author
[Sergei Khlebnikov]{Sergei Khlebnikov}
\email{skhleb@physics.purdue.edu}

\address{Department of  Physics, Purdue University,
 West Lafayette, IN 47907 and Institute for Advanced Study, Princeton NJ, 08540}

\author
[Birgit Kaufmann]{Birgit Wehefritz--Kaufmann}
\email{ebkaufma@math.purdue.edu}

\address{Department of Mathematics and Department of Physics, Purdue University}

\begin{abstract}
We study a class of graph Hamiltonians given by a type of quiver representation to which we can associate (non)--commutative geometries.   
By selecting gauging data these geometries are realized by matrices through an explicit construction or a Kan-extension.
We describe the changes in gauge via the action of a regauging groupoid.
It acts via matrices that give rise to a noncommutative 2--cocycle and hence
to a groupoid extension (gerbe).

We furthermore show that automorphisms of the underlying graph of the quiver can be lifted to extended symmetry groups
of regaugings.
In the commutative case, we deduce that the extended symmetries act
via a projective representation. This  yields isotypical decompositions and super--selection rules.

We  apply these results to the PDG and honeycomb wire--networks  using representation theory for projective groups and
show that all the degeneracies in the spectra are consequences 
of these enhanced 
symmetries. This includes the Dirac points of the G(yroid) and the honeycomb systems.

\end{abstract}

\maketitle


\section{Introduction}

We study a class of graph Hamiltonians given by a type of quiver representation to which we can associate (non)--commutative geometries.   
Our particular focus are symmetries in these geometric realizations especially those coming from the symmetries of the graph. Via considering a re--gauge group(oid) action we can show that the classical graph symmetries lead to enhanced (centrally extended) symmetries which are realized as projective representations in the commutative case.

The physical motivation for considering such systems stems from considering wire--systems on the nano--scale where the presence of higher dimensional irreps in the decomposition of the above symmetries leads to degeneracies in the spectrum. After giving the general arguments we apply them to the PDG and the honeycomb wire systems. Here we are especially motivated by understanding the electronic properties of a novel material \cite{Hillhouse} based on the G(yroid) geometry \cite{kkwk}.
We expect that our considerations can also be applied to other graph based setups, such as those coming from quiver representations, e.g.\ in field theory, 
or the coordinate changes in cluster algebras and varieties.

Mathematically, the initial data we start from is a finite graph $\BG$ together with a separable Hilbert space  $\H_v$ for each vertex $v$ of the graph and a unitary morphism for each oriented edge, such that the inverse oriented edge corresponds to the inverse morphism.
Algebraically this data corresponds to a groupoid representation in separable Hilbert spaces, as we explain in \S\ref{groupreppar}. 
In this situation, as we derive, there is an associated Hamiltonian acting on the direct sum of all the Hilbert spaces $\H_v$.
 
In order to obtain a matrix representation of the Hamiltonian one has to fix some additional gauge data.
The gauge data consists of a rooted spanning tree and an order on the vertices.
With this choice in place, each edge corresponds to a loop and
 we can represent an isometry associated to an edge by an element of the $C^*$--algebra $\A$  generated by the morphisms corresponding to the loops of $\BG$ at a fixed base point, cf.\  \S\ref{pipar}.  Via pull--back this also yields a matrix representation of the Hamiltonian  in $M_k(\A)$ where $k$ is the number of vertices of $\BG$.

From the noncommutative geometry point of view the $C^*$--algebra $\A$ represents a space. If $\A$ is commutative (since $\A$ is unital)
this space can be identified as a compact Hausdorff space $X$ such that the $C^*$--algebra of complex valued continuous functions $C^*(X)$ is isomorphic to $\A$. In the applications we consider $X$ is the momentum space, which  in the commutative case is the n--dimensional torus
 $X=T^n=(S^1)^{\times n}$ 
and in the noncommutative case $\A$ is the noncommutative $n$--torus $\TTheta^n$ for a fixed value of $\Theta$, that
in physical situations is given by a background  B--field. See below \S\ref{wiregeosec}.
 
Concrete extended symmetry groups are constructed via a lift of the action of the underlying graph symmetries $\group$ on this data as re--gaugings.
The lift of the classical symmetries is rather complicated and proceeds in several steps:

(1) We first establish that the different matrix realizations
of the Hamiltonian given by choosing
different rooted spanning trees and orders are all linked by gauge--transformations---see Theorem \ref{regaugethm}.
The specific gauge transformations that arise form the re--gauging groupoid $\Gpd$. It
acts transitively on the set of all the matrix Hamiltonians that can be obtained from the decorated graph by all different choices of data.
Using category theory these realizations are just Kan--extensions given by pushing forward to the graph obtained by contracting the spanning tree.

(2) We then show that the gauge transformations can be represented as conjugation 
with matrices with coefficients in $\A$. We prove  that these matrices lead to a noncommutative 2--cocycle.
This in turn gives rise to a groupoid extension of $\Gpd$. Geometrically this corresponds to a gerbe.

(3)  In the commutative case (\S\ref{compar}), we furthermore show that these matrices
give a {\em projective} representation of the re--gauging groupoid. 
Just like in ordinary theory
of projective representations this means that there is a {\it bona fide} representation 
of a central extension of this groupoid.

(4) In the commutative setup, if we fix a point $p\in X$ and evaluate
the matrix Hamiltonian with coefficients in $\A$ at $p$ we obtain a matrix Hamiltonian with coefficients in $\C$ that we denote by $H(p)\in M_k(\C)$. In this way we can think of $X$ as the base of a family of finite dimensional Hamiltonians. Likewise, the re--gauging actions give a groupoid representation in matrices $M_k(\C)$
which commute with $H(p)$.
The stabilizer groups of a particular fixed Hamiltonian are the sought after enhanced symmetry groups. 

(5) For  applications this leaves the problem of identifying the points $p$ and the stabilizer
groups or at least subgroups. To address the latter question, we establish  that the automorphism group  $Sym(\BG)$
of the graph  induces re--gaugings, by pushing forward the spanning tree and the order of the vertices.
In this way, the  symmetries of the graph give rise 
to a sub--groupoid of $\Gpd$. 
Going through the construction
outlined above, we can restrict to this sub--groupoid  and see that at a fixed point 
of the re--gauging action we get a projective representation of the stabilizer subgroup which leads to possible
higher dimensional irreps and thus band sticking.

(6) In order to identify points of $X$ ---which we take to be $T^n$ for concreteness--- where such enhanced symmetry groups can occur, we show that under certain assumptions, that hold in all cases of our initial physical interest, the operations of the  symmetry group of the graph $Sym(\BG)$ via re--gaugings lift to an action on the base torus $T^n$ (Theorems \ref{autothm} and \ref{nondegthm})\footnote{We say lift to the base here, since any action on the base gives rise to an action on the parameterized family of Hamiltonians, but it is not clear that any such action comes from one on the base.}.
At points  $t\in T^n$ with non--trivial stabilizer groups, we automatically get a projective representation of these stabilizer subgroups of the automorphism groups of the underlying graph, which commutes with the Hamiltonian. Hence
we get isotypical decompositions, which can give us non--trivial information about the spectrum using
the arguments  above.

We wish to point out that this approach is broader than that of considering classical symmetries of decorated graphs and
in the commutative case generalizes the extensive  analysis of \cite{Avron}, see \S\ref{classympar} for details.

We apply all these considerations to the cases of the PDG wire networks and the honeycomb lattice; see \S\ref{calcpar}. Here 
 the graph $\BG$ arises physically as the quotient
graph  of a given (skeletal) graph $\G\subset \R^n$ by a maximal translation group $L\simeq \Z^n$. Each edge of the quotient graph is decorated with a partial isometry operator of translation in the direction of any lift of that edge to the graph $\G$.
The Harper Hamiltonian is constructed from these isometries; in the simplest version
(tight-binding approximation), it is simply the summation over them.
This Hamiltonian along with the symmetries of the given material are the main input into the noncommutative geometry machine, which constructs a $C^*$--algebra that encodes relevant information about the system out of this data.

One main objective is  to analyze and understand
the branching behavior or  stated otherwise the locus of degenerate Eigenvalues.
The motivation is that in solid-state physics such degenerate Eigenvalues may lead to 
novel electronic properties, as is the case, for instance, with the Dirac points in graphene \cite{CastroNeto}.
The key observations are that (a) non--Abelian extended
symmetry groups  by themselves can force degeneracies via higher dimensional --i.e. $>1$-- irreducible representations
and (b) any enhanced symmetries, also Abelian ones, give rise to super--selection rules. The latter ones
can facilitate finding the spectrum considerably, since the Hamiltonian Eigenspace decomposition
has to be compatible with the isotypical decomposition of the representation.

Complementary to this group theoretic approach there is another
one  via singularity theory,
which is contained in \cite{kkwkdirac}. 
 Our main result for the PDG and honeycomb networks is that both approaches yield the same classification of degeneracies in the commutative case. 
 Namely, at all degenerate points, which were analytically classified in \cite{kkwkdirac},
there is an enhanced symmetry group arising from graph automorphisms in the above way which forces the degeneracy. Here the surprising
fact is that we find {\em all} the degeneracies and degenerate points through the projective representations of (subgroups of) $\group$ given by re--gaugings.

Here the G--wire network corresponding to the double Gyroid which was our original motivator is the most interesting case. As shown in \cite{kkwkdirac} there are exactly two points with triple degeneracies
and two points with degeneracy $(2,2)$, that is two doubly degenerate Eigenvalues.
The automorphism group of the graph is $\SS_4$. The representation theory becomes
very pretty in this case. There are two fixed points $(0,0,0)$ and $(\pi,\pi,\pi)$ 
on $T^3$ under the whole $\SS_4$ action.
The projective representation is just the ordinary representation of $\SS_4$ by $4\times 4$ permutation matrices at the point $(0,0,0)$ which is known to decompose into the trivial and a 3--dimension irreducible representation which forces the triple degeneracy. This result was also found by \cite{Avron}, where an initially  different system
was considered that results in the same spectrum.

 At the point $(\pi,\pi,\pi)$
things become slightly more interesting. There is a projective representation of $\SS_4$,
but we can show that this projective representation corresponds 
to an extension which is isomorphic
to the trivial extension. Hence after applying the 
isomorphisms we  have a representation of $\SS_4$ and  
there is a trivial and a 3-dim irrep, giving the second triple degeneracy. 
The elucidates the origin of the symmetry stated in  \cite{Avron}.
Notice that the classical symmetries of decorated graphs would only yield an $\SS_3$ symmetry 
at this point, which cannot
explain a triple degeneracy as there are no 3-dim irreps for this group.

Things really get interesting at the 
two points $(\pm \pi/2,\pm\pi/2,\pm \pi/2)$. Here the
stabilizer group is $A_4$. The projective representation gives 
rise to an extension which we
show to be isomorphic to the non--trivial double 
cover $2A_4$ of $A_4$ aka.\ 2T, the binary tetrahedral group or $SL(2,3)$. Using the 
character table we deduce that the representation decomposes into two 2--dim 
irreps forcing the two double degeneracies.

These are completely novel results. We wish to point
out that one absolutely needs the double cover as $A_4$ 
itself has no 2-dim irreps and hence the projective extension is essential.

We also use the fact that the diagonal of $T^3$ is fixed by a 
cyclic subgroup $C_3$ of $A_4$ in
order to determine the spectrum analytically. 
Here we use the super--selection rules. 

For the D and honeycomb case, we show that the 
degenerate points which are well known in the honeycomb case and were 
computed for D in \cite{kkwk2} are all detected by enhanced symmetries. 
These however yield Abelian representations and hence we have to use the arguments 
of the type (b), that is super--selection rules,
to show that the Eigenvalues are degenerate over these points. Similar results to ours have now also been independently found for the D case in \cite{Freed} using different methods.

\section{General setup}
In this section, we show how to construct the $C^*$--algebra $\A$ and the Hamiltonian $H\in M_k(\A)$ mentioned
in the introduction from a graph representation of a finite graph $\BG$ with $k$ vertices. Furthermore we  embed a copy $\A$ into  $M_k(\A)$  and define $\B$ to be the subalgebra generated by $H$  and $\A$ under this embedding. The pair $\A\hookrightarrow \B$ is the basic datum for our noncommutative geometry.

\subsection{Groupoid graph representations in separable Hilbert spaces}
\label{groupreppar} Given a finite graph 
$\BG$ we define a groupoid representation of $\BG$, an association of a separable Hilbert space $\H_v$ for each $v\in V(\BG)$ and an isometry $U_{\vec{e}}:\H_v\to \H_w$ 
for each directed edge $\vec{e}$ from $v$ to $w$. 

This data indeed determines a unique 
functor $\rep$ from  the path groupoid of $\BG$ to the category of separable Hilbert spaces.
The path groupoid ${\mathcal P}_{\BG}$ (or $\P$ for short) of $\BG$ is category whose objects are the vertices
of $\BG$ and whose morphisms are generated by the oriented edges, where the inverse
of a morphism given by $\vec{e}$ is the one given by $\cev{e}$.  
Notice, that we are  looking at the morphisms generated by the oriented edges,
this means that $Hom_{\mathcal P}(v,w)$ is the set of paths along oriented edges from $v$ to $w$ modulo
the relation that going back and forth along an edge is the identity.
Composition is only allowed if the first
oriented edge terminates at the beginning of the second oriented edge. This is why we only obtain a groupoid and not a group.

In particular, this lets one view the  fundamental group $\pi_1(\BG,v_0)$ in two equivalent fashions. First as
the topological $\pi_1$ of the realization of a graph, and secondly as 
the group $Hom_{\mathcal P}(v_0,v_0)=Aut(v_0)$ where $Hom_{\mathcal P}$ are the morphisms in the path groupoid ${\mathcal P}$. 

The collection of automorphisms in $\P$ forms a subgroupoid $\Aut$. It is the disjoint union $\Aut=\amalg_{v\in V(\BG)} \pi_1(\BG,v)$.
These are the classes of free loops on $\BG$.

\subsubsection{Hamiltonian, symmetries and the  $C^*$--geometries}

Given a groupoid representation as above, set $\H=\bigoplus_{v\in V(\BG)}\H_v$
and define $H$ by 
\begin{equation}
\label{haperhameq}
H=\sum_{e\in E(\BG)} \left( U_{\vec{e}}+U_{\stackrel{\leftarrow}{e}} \right) \in B(\H)
\end{equation}
where $B(\H)$ is the $C^*$ algebra of bounded operators on $\H$.

We let $\A$ be the {\em abstract} $C^*$--algebra $\pi_1(\BG,v_0)$ generates in $B(\H)$ via the representation.
This is a bit subtle, as the concrete algebra depends on the choice of base point $v_0$. We will use the notation $\A_{v_0}:=U(\pi_1(\BG,v_0)$ to emphasize 
this.\footnote{Here and below for any subgroupoid  $\P'$ of $\P$ we denote the $C^*$--subalgebra of $B(\H)$ generated by the morphisms of $\P'$ via $U$  by $U(\P')$.}

Of course any two choices of a base vertex give isomorphic algebras
but there is no preferred isomorphism between them. In the physical situation of wire networks, we are interested
in \S\ref{wiregeosec}, there is however a global identification of these algebras which comes from the embedding of the system into 
Euclidean space. 
Algebraically we realize this as extra coherence  isomorphisms
$\coh_{*v}:\A_{v}\stackrel{\sim}{\to} \A$ with inverse $\coh_{v*}:=\coh_{*v}^{-1}$.

The direct sum of the $\coh_{v*}$ gives a representation $\symrep$ of $\A$ into $\Atot:=U(\Aut)=\bigoplus_{v\in V(\BG)}\A_{v}\subset B(\H)$. 
The algebra $\B$ is now the sub $C^*$--algebra generated by $H$ and $\symrep(\A)$. We also set $\Btot=U(\P_{\BG \, 1})\subset B(\H)$. 
 
The $C^*$--geometry we are interested in is the inclusion $\A\to\B$. We call the system commutative, if $\B$ (and hence $\A$) is commutative.
We call the situation {\em fully commutative} if in addition for any all $v$ and $w$  and any
contractible edges path $\gamma$ from $v$ to $w$ $\alpha_{*v}U_{\gamma}\alpha_{w_*}=id$.

 In the wire--network case, the condition to be fully commutative corresponds to the case of zero magnetic field. 
 
 {\sc Notation:} If we choose a fixed base point $v_0$ we will tacitly use the isomorphism $\alpha_{v_0*}$ to identify $\A$ and $\A_{v_0}$.

 \subsubsection{Matrix Hamiltonian} 
 \label{gammapar}
If we   
fix a rooted spanning tree and
an order on the vertices, we can identify $\H\simeq \H_{v_0}^k$ via
a unitary $U$ as follows and $H$ becomes equivalent to a matrix in $M_k(\A)$. 

A spanning tree $\t$ is by definition a contractible subgraph of $\BG$ which contains all the vertices
of $\BG$. It is rooted if one of the vertices is declared the root; denote it by $v_0$. 
In order to obtain a honest $k\times k$ matrix, we 
also have fix an order $<$ on all the vertices, where we insist that the root 
is the first vertex in this order.  That is we fix a rooted ordered spanning tree $\bt:=(\t,v_0,<)$.

For each vertex $v$ of $\BG$ there is a unique shortest path in $\t$ to $v_0$. This defines a choice of fixed isomorphism
$U^{\tau}_{vv_0}:\H_{v_0}\to \H_v$ by translations along the edge path, see \S\ref{isopar} for details. 
Assembling these maps gives the desired isomorphism $\H\simeq \H_{v_0}^k$. Pulling back $H$ to $ \H_{v_0}^k$ using this isomorphism, we obtain {\em a} matrix version $H_{\bt}$ where we include
the subscript to stress that this matrix depends on the choice of rooted ordered spanning tree $\bt$.

To fix the notation, which we will need later, we give the full details: let $v_i$ be the i--the vertex in the enumeration $<$.
Then we obtain a matrix $H_{\bt}$ by using the isomorphisms $U^{\tau}_{v_iv_0}$
which are defined as follows. Let $v_0,w_1,\dots,w_k,v_i$ be the sequence
of vertices along the unique shortest path $\gamma^{\t}_{v_iv_0}$ from $v_0$ to $v_i$ in $\tau$,
then
\begin{equation}
\label{Utaudefeq}
U^{\tau}_{v_iv_0}=U_{v_iw_k}U_{w_kw_{k-1}} \dots U_{w_2w_1}U_{w_1v_0}=U(\gamma^{\t}_{v_iv_0})
\end{equation}
 and
$U^{\tau}_{v_0v_i}=(U^{\tau}_{v_iv_0})^*$. Given the choice of $\bt$, we get the corresponding
matrix Hamiltonian as
\begin{equation}
H_{\bt}:\bigoplus_{i=1}^k \H_{v_0}\stackrel{\bigoplus_iU^{\tau}_{v_0v_i}}{\longrightarrow}
\bigoplus_i \H_{v_i}=\H \stackrel{H}{\longrightarrow}
\H=\bigoplus_i \H_{v_i}
\stackrel{\bigoplus_iU^{\tau}_{v_iv_0}}{\longrightarrow}
\bigoplus_i \H_{v_0}
\end{equation}

Of course, all the $\H_{\bt}$ are equivalent, although not canonically. For the equivalence one
has to choose a path from one root to the other. We will exploit this fact extensively below.

\subsubsection{$\A$ weighted graph}
After having fixed an initial spanning tree the matrix Hamiltonian has a different description.
To each edge $\vec{e}$ from $v$ to $w$, we can associate $\wt(\vec{e}):=U^{\tau}_{v_0v}U_{\vec{e}}U^{\tau}_{vv_0}\in U(\A)$.
That is, we can regard $\BG$ as having a weight function on ordered edges with weights in $U(\A)$.
If $e$ is an edge of $\tau$ then with this definition $\wt(\vec{e})=\wt(\cev{e})=1$.

An alternative way of viewing this data is as a certain type of quiver representation, we will comment on this more below.

\subsubsection{Weights as a representation of the fundamental group}
\label{pipar}
For any finite  graph $\BG$ the Euler Characteristic is $\chi(\BG)=|V(\BG)|-|E(\BG)|=1-b_1$ where $b_1$ is the rank of $H_1(\BG)$ which
is the same as the ``number of loops''. More precisely,  if $\BG$ is connected,  $\pi_1(\BG)={\mathbb F}_{b_1}$ that is the free group in $b_1$ generators. In the applications  $b_1$ is the rank of  the lattice of translational symmetries.

One way to view a rooted spanning tree ($\t,v_0)$ is to think of it as fixing a base point $v_0$ and a set of symmetric generators/basis for $\pi_1(\BG,v_0)=Hom_{\mathcal P}(v_0,v_0)$. Topologically after contracting the spanning tree one is left with a wedge of  $S^1$s. There are $b_1$ of
these, one for each non--contracted edge.
Each simple loop around one of the $S^1$s gives a generator. Picking  one generator per loop gives a basis.

Without doing the contraction the correspondence on $\BG$ itself is 
 given by all the (ordered) edges not contained in $\t$. To each such ordered edge $\vec{e}$ from $v$ to $w$, we associate the loop 
starting at $v_0$ going along $\t$ to $v$ then traversing $\vec{e}$ to $w$ and afterwards returning to $v_0$ along $\t$.
Again picking both orientations gives a symmetric set of generators of the free group while picking only one
orientation per edge fixes a basis. Any edge in the spanning tree corresponds to the unit, that is the class of a constant loop.

In this language, the weight function $\wt$ is a representation of $\pi_1$ lifted to the
edges of the graph by the above correspondence.
Thus, as long as the base point $v_0$ stays fixed, the changes of spanning tree  can be viewed as a change--of--basis
of $\pi_1(\BG,v_0)$.
If $v_0$ moves, say to $v_0'$, then, as usual, any path from $v_0'$ to $v_0$ 
gives an isomorphism  taking $\pi_1(\BG,v_0')$ to $\pi_1(\BG,v_0)$.
Both types of isomorphisms will play a role later in the symmetry group actions.

\subsubsection{Non--degeneracy and toric non--degeneracy}
\label{nondegsec}
We call a groupoid representation non--degenerate, if the images of the basis of the free group
given by the construction above are independent unitary generators 
of $\A$ and call it 
toric non--degenerate if $\A$
is isomorphic to the noncommutative torus $\TTheta^{b_1}$.

Notice that if $\A$ is commutative and non--degenerate, then $\A\simeq \T^{b_1}$, the $C^*$ algebra of the 
torus of dimension $b_1$.

\subsubsection{Hamiltonian and $\A$ from a weighted graph}
\label{graphhamsec}
Alternatively to starting with a groupoid representation 
one can also start with an $\A$--weighted graph. It is in this representation that we can understand the re--gauging groupoid $\Gpd$.

Fix a finite  connected graph $\bar \Gamma$,
a rooted ordered spanning tree $\bt$ of $\bar \Gamma$ such that
 the root of $\tau$ is the first vertex, a unital $C^*$ algebra $\A$, and a morphism
$\wt:\{\text{Directed edges of } \G\}\to \A$ which satisfies
\begin{enumerate}
\item $\wt(\vec{e})=\wt(\cev{e})^*$
if $\vec{e}$ and $\cev{e}$ are the two orientations of an edge $e$.
\item $\wt(\vec{e})=1\in \A$ if the underlying edge $e$ is in the spanning tree.
\label{normcond}
\end{enumerate}
In general, if $\wt$ is as above and it satisfies the first condition, we will call it a weight function (with values in $\A$)
and if it satisfies both conditions, a weight function compatible with the spanning tree.

Fix a separable Hilbert space $\H_{v_0}$. By Gel'fand Naimark representability we realize
$\A\subset B(\H_{v_0})$ and we shall use this representation. 

We shall also postulate that 
{\em $\A$ is minimal}, which means that it is the $C^*$--algebra generated by the $\wt(\vec{e})$ where $\vec{e}$ 
runs through the directed edges of $\BG$. This makes the terminology of \S\ref{nondegsec} applicable.
Also, we see that this is again just a lift of a representation of $\pi_1(\BG,v_0)$ to the edges of $\BG$
using the spanning tree $\t$.

Given this data, let $k$ be the number of vertices of $\Gamma$. We will enumerate
the vertices $v_0,\dots, v_{k-1}$ according to their order; $v_0$ being the root.
Given this data, the Hamiltonian $H=H(\Gamma,\bt,w)$ is the matrix $H=(H_{ij})_{ij}\in M_k(\A)$ 
whose entries are:
\begin{equation}
H_{ij}=\sum_{\text{directed edges $\vec{e}$ from $v_i$ to $v_j$}} \wt(\vec{e})
\end{equation}
 It acts naturally on
 $\H:=\H_0^k$.
In this sense, the weighted graph encodes  both the Hamiltonian {\em and} the symmetry algebra $\A$.

In the general noncommutative
case, this is {\em not quite enough} for the whole theory, as we do not recover the action of $\A$ on $\H$ and
the connection between the action of $H$ and that of $\A$. 
 Recall that the action of $\A$ on $\H=\H_{v_0}^k=\bigoplus_{v\in V(\BG)} \H_{v_0}$
is given on each summand $\H_{v_0}$ corresponding to $v$ by pulling back the action from $\H_v$.
That is, the true action is a conjugated action.  

In the {\em commutative case} this is not an issue as the representation is exactly the diagonal representation.

\subsubsection{Geometry in the commutative case}
\label{comgeopar}
If $\B$ is commutative (and hence also $\A$),\footnote{ This is for instance the case in the applications if the magnetic field vanishes.} then there is a geometric version 
of these algebras which can be understood as  the spectra of a family of Hamiltonians over a base.
We have the following inclusion of commutative $C^*$ algebras $i:\A\hookrightarrow \B$,
by Gel'fand Naimark this gives us a surjection of compact Hausdorff spaces\footnote{Both $\A$ and $\B$ are unital.} 
$\pi:Y\to X$ where $C(X)\simeq \A$ and $C(Y)\simeq \B$. 
The correspondence is given via characters. Namely a character is a $C^*$--homomorphism 
$\chi:\A\to \C$ . The characters are by definitions the points of $X$. Vice--versa
any point $t\in X$ determines a character $ev_t:C^*(X)\to \C$ via evaluation. That is any $f\in C^*(X)$
is sent to $f(t)\in C$. Given a character $\chi$ on $\A$, we can lift it to a $C^*$--morphism
 $\hat \chi:M_k(\A)\to M_k(\C)$ by applying it in each matrix entry.
 
 Thus any point $t\in X$ represented by the character $\chi$ determines a Hamiltonian 
 $\hat \chi(H)\in M_k(\C)$ via $\hat\chi$. 
\begin{equation}
(\hat \chi(H))_{ij}=\chi(H_{ij})
\end{equation}

Thus we get a family of Hamiltonians $H(t)$ parameterized over the base.
One can furthermore check, see \cite{kkwk}, that $\pi$ is a branched cover over $\A$ with 
$\pi^{-1}(t)=spec(H(t))$.

\subsection{Physical example: PDG and honeycomb wire networks}
\label{wiregeosec}

The PDG examples are based on the unique triply periodic CMC surfaces
 where the skeletal graph is symmetric and self--dual.
Physically, in the P (primitive), D (diamond) and G (Gyroid) case, one starts with a ``fat'' or thick version of this surface,
which one can think of as an interface. A solid-state realization of the ``fat'' Gyroid aka.~ double Gyroid has recently been synthesized 
on the nano--scale \cite{Hillhouse}. 
The structure contains three components, the ``fat'' surface or wall and two channels.
Urade et al.\ \cite{Hillhouse} 
have also demonstrated a nanofabrication technique in which the channels are 
filled with a metal, while the silica wall can be either left in place or removed. This yields
two wire networks, one in each channel. The graph we consider and call Gyroid graph
 is the skeletal graph of one of these channels.  
 The graph Hamiltonian we construct algebraically below is the tight-binding
Harper Hamiltonian for one channel of this wire network.
 The 2d analogue is the honeycomb lattice underlying graphene. 
Graph theoretically the quotient graph for the honeycomb is the 2d version of that of the D surface, but as we showed, the behavior, such as the existence of Dirac points \cite{kkwkdirac}, is more like that of the Gyroid surface. 

To formalize the situation, we regard the skeletal graph of one channel as
 an embedded graph  $\G\subset \R^n$. The crystal structure  gives a  maximal translational symmetry group $L$ which is a 
mathematical lattice, i.e.\ isomorphic to $\Z^n$, s.t.\ $\BG:=\G/L$ is a {\em finite} connected graph. The vertices of
this quotient graph are the elements in the primitive cell.

The Hilbert space $\H$ for the theory is $\ell^2(V(\G))$, where $V(\G)$ are the vertices of $\G$.
This space splits as $\bigoplus_{v\in V(\BG)} \H_v$ where for each vertex $v\in V(\BG)$, $\H_v=\ell^2(\pi^{-1}(v))$
where $\pi:\G\to \BG$ is the projection. All the spaces $\H_v$ are separable Hilbert spaces and hence
isomorphic. Furthermore if $\vec{e}$ is a directed edge from $v$ to $w$ in $\BG$ then it lifts uniquely as a vector to $\R^n$,
which we denote by the same name. Moreover for each such vector there is a naturally associated translation operator 
$T_{\vec{e}}:\H_v\to \H_w$, by the usual action of space translations on functions. 
We also allow for a constant magnetic field $B=2\pi \hat \Theta$ where
 $\hat \Theta=\sum \theta_{ij} dx_i\wedge dx_j$ is  a constant 2--form given by the skew--symmetric matrix $\Theta=(\theta_{ij})_{ij}$.
If $\Theta\neq 0$   then the translations become magnetic translations or Wannier operators $U_{\vec{e}}:\H_v\to \H_w$; see e.g.\
 \cite{bellissard}.
These operators are still unitary and give partial isometries when regarded on $\H$ via projection and inclusion, which
we again denote by the same letter.
The Harper Hamiltonian is then defined by equation (\ref{haperhameq}).

Likewise $L$ acts by magnetic translations. If $\vec{\lambda}$ is a vector in $L$ then
$U_{\vec{\lambda}}$ sends each $\H_v$ to itself and the diagonal action gives an action on $\H$. Since $L$ is a lattice,
the representation it generates is given by $n$ linearly independent unitaries. The commutation relations
among the Wannier operators amount to the fact that the representation is a copy of $\TTheta^n$,
 the noncommutative $n$--torus, see \cite{kkwk} or \cite{bellissard}. The direct sum yields the global symmetry representation $\alpha:\TTheta^n\to B(\H)$.
This representation and $H$ by definition generate the Bellissard-Harper $C^*$--algebra $\B$.

We see that this is an example of the general setup where the groupoid representation of $\BG$ is given by the $U_{\vec{e}}$. Here
$\A=\TTheta^n$ and the Hamiltonian associated to the graph is the Harper Hamiltonian. If $B=0$ then the situation is fully commutative.
Notice that $\A$ being commutative just means that the fluxes though the $L$ lattice are $0$. If $\B$ is commutative, this means
the the fluxes cancel according to the entries of $H$. If there are vertices with more than one edge, potentially the situation is commutative, but
not fully commutative.

\section{Symmetries}
In order to deal with the symmetries of the graph it is helpful to first fix the notation as
sometimes questions become subtle.

\subsection{Classical Symmetries}
\label{classympar}
A graph $\G$ is  described by a set of vertices $V_{\Gamma}$ and a set of edges $E_{\G}$ together with incidence
relations $\del$ where for each edge $e$, $\del(e)=\{v,w\}$ is the unordered set
of the two vertices it is incident to. A directed edge is given by an order on this pair.
Hence for each edge $e$ there are two ordered edges  by the orders $(v,w)$ and
 by $(w,v)$. We usually denote these two edges by $\vec{e}$ and $\cev{e}$.
 The set of all oriented edges is called $E^{or}_{\Gamma}$.

 An isomorphism $\phi$ of two graphs $\G$ and $\G'$ is a pair of bijections $(\phi_V,\phi_E)$
 $\phi_V:V_{\G}\to V_{\G'}$ and $\phi_E:E_{\G}\to E_{\G'}$. The compatibility is that the incidence conditions are preserved, i.e.: if $\del(e)=\{v,w\}$ then
 $\del(\phi_E(v))=\{\phi_V(v),\phi_V(w)\}$. Notice the $\phi$ also induces a map
 of oriented edges, the orientation of the edge $\phi_E(e)$ given by $(\phi_V(v),\phi_V(w))$, if the orientation of $e$ is $(v,w)$.

 We will treat isomorphism classes of graphs from now on.  Fixing an isomorphism class of a graph still allows for automorphisms. These are given as follows.
 Fix a representative of $\Gamma$ then an automorphism is a pair of compatible maps $(\phi_V,\phi_E)$;
 $\phi_V:V_{\G}\to V_{\Gamma}$ and $\phi_E:E_{\G}\to E_{\G}$.

 \begin{ex}
 Let us illustrate this for the  graphs corresponding to the PDG and honeycomb cases
 which are given in Figure \ref{graphsfig}; see \S\ref{wiregeosec} for details about the corresponding wire networks. 
 We fix once and
 for all an isomorphism  class of these graphs and then consider their automorphisms using the representatives given in the figure.

\begin{figure}
\includegraphics[width=.8\textwidth]{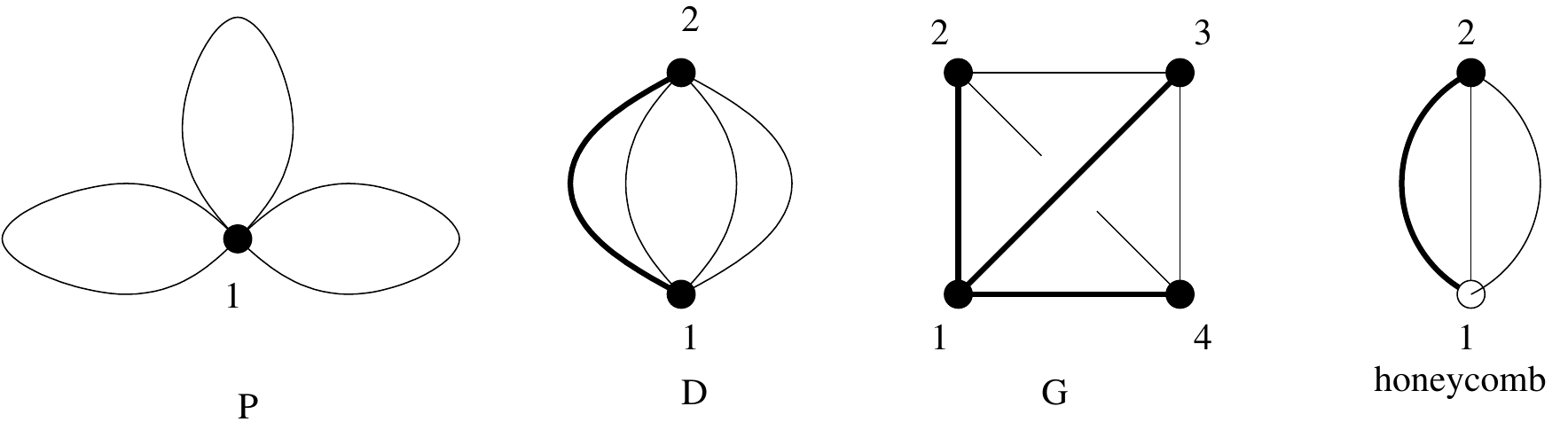}
\caption{The graphs P, D and G and honeycomb together with the preferred spanning tree and order.}
\label{graphsfig}
\end{figure}

 For the P case there is only one vertex hence $\phi_V=id$ is the only possibility. However there is an $\SS_3$ action permuting the three loop edges.

 The D graph has the possibility of switching the two vertices and freely permuting the three edges. This gives the automorphism group $\Z/2\Z\times \SS_3$.
 The honeycomb similarly has automorphism group $\Z/2\Z\times \SS_3$.

 For the Gyroid, there is an $\SS_4$ worth of potential choices for $\phi_V$. Now all these choices extend uniquely to the edges, since there is exactly one edge between
 each distinct pair of vertices and hence the symmetry group is exactly $\SS_4$.
\end{ex}

\subsubsection{Pushing forward Spanning Trees}
Given a pair $(\G,\t)$ of a graph and a rooted spanning tree, we define the action of isomorphisms
and automorphisms by push--forward.
That is an isomorphism between $(\G,\t)$  and $(\G',\t')$ is an isomorphism from $\G$ to $\G'$ such that
that $\phi_V$ maps the root of $\tau$ to the root of $\tau'$ and $\phi$ restricted to $\tau$ is
an isomorphism onto $\tau'$.

If we have not already specified a spanning tree on $\G'$ we can extend any isomorphism  $\phi$ from $(\G,\t)$ to it
by push--forward.
 This means that we push--forward all the vertices
and the edges of the spanning tree $\t$ to $\G'$: $E_{\tau'}:=\phi_E(E_{\tau})$ and likewise push--forward the root.

In particular, $Aut(\G)$ acts on the set of spanning trees of a fixed graph $\G$.
This action is not transitive in general and may have fixed points.

\begin{ex}
In the cases of PDG and the honeycomb, it is a transitive action.

For the G graph the action is not fixed point free,
 there is an $\SS_3$ subgroup fixing a given spanning tree.

For the P graph the action is  fixed point free, while for the D and the honeycomb
although the action is transitive, there are again stabilizers. For the honeycomb the group fixing
a spanning tree is the $\SS_2=\Z/2\Z$ interchanging the two other edges, with both vertices fixed,
while in the D case it is an $\SS_3$ action interchanging the edges which are not part of the spanning tree.

As an example of a non--transitive action consider the triangle graph, with one edge doubled. That is three vertices $1,2,3$ with one edge between $1$ and $2$, one edge between $1$ and $3$ and two edges between $2$ and $3$. 
\end{ex}

\subsubsection{Orders, Weight Functions and Isomorphisms}
If there is an order on all the vertices then the isomorphisms are asked to be compatible with this
order and auto-- and isomorphisms can be extended by 
pushing forward the order.

A final piece of data on a graph, which we utilize is the $C^*$ algebra valued map $\wt:E^{or}_{\Gamma}\to \A$,
such that $\wt(\cev{e})=\wt(\vec{e})^*$.
In our case of interest $\A=\TTheta^n$, and even $\T^n_0$ later, but for now we will keep the abstract setting.

Notice that a given weight function might not be 
compatible with the new, pushed--forward spanning tree, 
since it does not necessarily satisfy the condition (\ref{normcond}).
The idea of re--gauging is to re--establish condition (\ref{normcond}) 
by passing to an equivalent, re--gauged
weight function.

One way to see the failure to preserve the condition (2) of \S\ref{graphhamsec} 
is to view the change of spanning tree
as a change of basis of $\pi_1$ or rather an isomorphism of $\pi_1$s
via the path--groupoid. This point of view is at the basis of the proof of Theorem
\ref{autothm} below.

\subsubsection{Classical vs.\ extended symmetries}
\label{classicalpar}
There is another
 natural choice of iso-- or automorphism for graphs with weight functions.
 Here one would postulate that the weight functions are compatible.
Weight functions  naturally pull back via $\phi_E$,
that is $\phi_E^*(w)(\vec{e})=\wt(\phi_E(\vec{e}))$. Using that $\phi_E$ is an isomorphism,
one can push--forward by pulling back along $\phi_E^{-1}$.

One could call these symmetries {\em classical symmetries of the weighted graph.} These are the kinds of symmetries that were for instance considered by \cite{Avron}. 

We will consider symmetries of the {\em underlying graph}, not of the weighted graph.
The weights are taken care of by re--gauging. One way to phrase this is that we utilize an extended symmetry group
which allows for phase factors at the vertices. The details are given below.

 \subsection{Gauging}
 \label{isopar}
We will now consider the relationship between the different matrices $H_{\bt}$ and $H_{\bt'}$ for different gauging data,
which  represent $H$
 via different isomorphisms. Here and in the following we use the notation: $\bt=(\t,v_0,<),\bt'=(\t',v_0',<')$, etc.
 
 There are basically three situations:
First, $\tau=\tau'$ as rooted trees and only
the order changes. Secondly, $\tau=\tau'$ but $v_0\neq v_0'$ and
thirdly, the trees just do not coincide, i.e.~ at least one edge is different.
 In the first  case the isomorphisms are simply changed by permutations
of the factors $\H_{v_0}$. This means that
the difference between the two matrix Hamiltonians is simply conjugation
by $M_{\sigma}$, the standard permutation matrix.
  Since $v_0$ is fixed to be the first element, the permutation actually lies
in the subgroup $\SS_{k-1}\subset \SS_k$ fixing the first element.
 $\SS_{k-1}$ then acts simply transitively
on the orders.

In the second and third case, the situation is more complicated. Of course the
second type of change is related to an action of $\SS_k$, but things are not
that simple, since there is a change of the space that $H_{\bt'}$ acts on.
If the tree moves, then we have to also change the weight functions to make them
compatible with the new spanning tree. 

Taking the point of view of $\pi_1$'s the first type of transformation is just a permutation of the basis.
But when we move the base point,
we  move to an isomorphic group. In doing this, we effectively use the path groupoid and not just the fundamental group.

\subsubsection{Gauging in groupoid representations}
As discussed in \S\ref{groupreppar} the Harper Hamiltonian {\em before} fixing a spanning tree
can be thought of as a certain type of groupoid representation. 

For such  representation,
we can re--gauge it to an equivalent representation  by acting with any of choice automorphisms of the $\H_v$,
that is the group $\times_{v\in \G}Aut(\H_v))$. Picking an element $\phi$
in this group is the same as the assignment $v\mapsto \phi(v)\in Aut(\H_v)$.
The operators $U_{\vec{e}}:\H_v\to \H_w$, where $\vec{e}=(v,w)$
get re--gauged to $\phi(w)U_{\vec{e}}\phi^{-1}(v)$. Again, one has to be careful with the
indexing of the direct sums. Since there is no natural order, there is a natural $\SS_k$ action
by permutations this interacts with the diagonal re--gaugings via the wreath product.

In our situation, since we have Hilbert spaces, we can look at unitary equivalences
and restrict the automorphisms to be unitary.
Note that the gauge group is smaller than the full group of unitary equivalences $U(\H).$

Also, choosing an identification of all the isomorphic separable Hilbert spaces
$\H_v$ with some fixed $\H_{v_0}$ we can take the re--gaugings to live in the unitary operators
on $\H_{v_0}$.

In this situation,  the gauge group becomes $G=U(\A)^k\wr \SS^k$. It acts
on the orders, the weight functions and on the Hamiltonians by conjugation and permutation just as above.

\subsubsection{Spanning tree re--gauging}

\begin{prop}
\label{regaugethm} Given two ordered rooted spanning trees $\bt$ and $\bt'$ there
is a matrix $M\in M_k(\A)$ with $MM^*=M^*M=id$ such that
$MH_{\bt}M^*=U^{\tau}_{v_0v'_0}H_{\bt'}U^{\tau}_{v'_0v_0}$. Moreover $M$ is an element of the gauge group. 
\end{prop}

\begin{proof} Consider the commutative diagram :
\begin{equation}
\xymatrix{
\bigoplus_{i=1}^k \H_{v_0}\ar[r]^{\bigoplus_iU^{\tau}_{v_0v_i}}\ar[d]^{\bigoplus U^{\tau}_{v_0v'_0}}&
\bigoplus_i \H_{v_i}=\H\ar[r]^{H}\ar[d]^{\sigma}&
\H=\bigoplus_i \H_{v_i}\ar[r]^{\bigoplus_iU^{\tau}_{v_iv_0}}\ar[d]_{\sigma}&
\bigoplus_i \H_{v_0}\ar[d]\ar[d]^{\bigoplus U^{\tau}_{v_0v'_0}}\\
\bigoplus_{i=1}^k \H_{v'_0}\ar[r]^{\bigoplus_iU^{\tau'}_{v'_0v'_i}}&
\bigoplus_i \H_{v'_i}=\H\ar[r]^{H}&
\H=\bigoplus_i \H_{v'_i}\ar[r]^{\bigoplus_iU^{\tau'}_{v'_iv'_0}}&
\bigoplus_i \H_{v'_0}
}
\end{equation}

We see that if $i'=\sigma(i)$ and $j'=\sigma(j)$ so that $v'_{i'}=v_{i}$:
\begin{eqnarray*}
&&(H_{\bt'})_{i'j'}=U^{\tau'}_{v'_0v'_{i'}}H_{v'_{i'}v'_{j'}}U^{\tau'}_{v'_{j'}v'_0}\\
&=&U^{\tau}_{v_0'v_0}(U^{\tau}_{v_0v'_0}U^{\tau'}_{v'_0v'_{i'}}
U^{\tau}_{v_iv_0}) (U^{\tau}_{v_0v_i} H_{v_iv_j}
U^{\tau}_{v_jv_0}) (U^{\tau}_{v_0v_j}
U^{\tau'}_{v'_{j'},v'_0}U^{\tau}_{v'_0v_0})U^{\tau}_{v_0v'_0}\\
&=&U^{\tau}_{v'_0v_0}\phi^*_{i'}(H_{\bt})_{ij}\phi_{j'} U^{\tau}_{v_0v'_0}
\end{eqnarray*}
With $\phi_{j'}= U^{\tau}_{v_0v_j}
U^{\tau'}_{v'_{j'},v'_0}U^{\tau}_{v'_0v_0}\in U( \A)$.
So that if $\Phi=diag(\phi_{i'})$ and  $M_{\sigma}$ is the permutation matrix of $\sigma$ which moves the order $<$ to $<'$ then $M=\Phi^* M_{\sigma^{-1}}$ and
 $M^*M=id$.
\end{proof}

We choose to place $M$ on the left of the Hamiltonian so as to get a left action later on. 

\begin{rmk} 
Unraveling the definition given in equation (\ref{Utaudefeq}),
we can express the matrix $\Phi$  as a re--gauging by the following iterative procedure.
We start at
the root of $\tau'$ and choose $\phi(v'_0)=id$. Assume we have
already assigned weights to all vertices at distance $i$  from $v'_0$ and let $w$ be a
vertex at distance $i+1$. Then there  a unique $v$ at distance $i$ which is connected to $w$ 
along a unique directed edge  $\vec{e}$ of the spanning tree $\tau'$. 
Set $\phi(w)=\wt(\vec{e})\phi(v) \in U(\A)$.
Then $\Phi=diag_{v_i'\in \tau'}(\phi(v_i'))$.
\end{rmk}

Of course the form of $M$ depends on the initial choice of $\phi(v_0)=id$,
 which amounts to using the iso $U^{\tau}_{v_0v'_0}$ to pull--back the matrix.
Any other choice of iso will differ by an element of $\A$ which
is then the value of $\phi$ on $v_0$. This plays a crucial role later.

\subsubsection{Commutative case. Reduced gauge group.}

In the commutative case, we can fix a character $\chi:\A\to \C$ and then under $\hat\chi$  all matrices become
$U(1)$ valued and all the Hilbert spaces $\H_v$ 
become identified with $\C$ . In this case, we can identify the gauge group action with an action
of
 $U(1)^{\times V_{\G}}$ on $U(1)$--valued  weight functions, using $\lambda=\chi\circ\phi$. For every oriented
 edge $\vec{e}$ from $v$ to $w$ the re--gauged weights are

 $$\wt'(\chi(\vec{e})):=\lambda(v)\chi(\wt(\vec{e}))\bar \lambda(w)$$

Notice that  we have taken the indexed un--ordered product. If we fix an order of the vertices,
then the group $\SS_k$ acts on the vertices as well and the full gauge group which acts on the Hamiltonians
by conjugation is
the wreath product $G=U(1) \wr \SS_k$. 

 We see that the constant functions $\lambda$ act trivially
and hence to get a more effective action we can quotient by the diagonal $U(1)$ action
and consider the reduced gauge group $\bar G:=G/U(1)$, where $U(1)$ is diagonally
embedded in $U(1)^k$ and $\SS_k$ acts trivially.

Abstractly $U(1)^k/U(1)\simeq U(1)^{k-1}$, 
to make this explicit,
we can choose a section of $\bar G\to G$.  Our choice
$\phi(v_0)=1$ is  just such a choice of a section. The action of $\SS_k$ on the remaining $k-1$ factors
is then more involved, however. It is still a semi--direct product, but not a wreath product any more. This has practical relevance in the Gyroid case.

The proof of the theorem above then boils down to the fact that
a rooted spanning tree uniquely fixes a unique gauge transformation as follows.
We let $\lambda(root)=1$ by the global gauge $U(1)$. Now 
 the weight on each vertex of the  tree is fixed  iteratively by the
 condition that $\lambda \wt (e)=1$. The whole set of weights then
gives a diagonal unitary matrix and taking the product with the appropriate permutation,
we obtain the matrix $M$.

\subsection{Re-gauging groupoid $\Gpd$, representations, cocycles and extensions}
In order to keep track of all the re--gaugings and ultimately find the extended symmetries, we introduce the following abstract groupoid 
$\Gpd$. It has rooted ordered spanning trees $\bt$
of $\BG$ as objects and a unique isomorphism between any two such pairs. If the two pairs
coincide, the isomorphism is the identity map.

Having fixed the representation $\rep$, 
there is an induced representation $\hrep$ of $\Gpd$ which also takes values in separable Hilbert spaces.
On objects it is given by  $\hrep(\bt)=\H_{v_0}^k$, $v_0$ being the base point of $\t$. For a re-gauging morphism $g:\bt\to \bt'$
we set $\hrep(g)=U^{\t}_{v'_0v_0}M$ for the $M$ of Proposition \ref{regaugethm}.  Plugging into the definitions one checks that indeed $\hrep(g)\hrep(h)=\hrep(gh)$
for composable $g$ and $h$.

In order to find the symmetry groups, we will however need to consider only the matrix ``$M$'' part of $\rho$. This is not a representation,
but gives rise to a noncommutative 2--cocycle and moreover, this cocycle can be lifted to the groupoid level.

\subsection{Induced structures and cocycles}
To understand the cocycle, let us first 
consider the ``$U$''--part of $\rho$. For this we notice that there is a functor $p:\P_{\Gpd}\to \P_{\BG}$ from the path space of $\Gpd$ to that of $\BG$.
It is given by $p(\bt)= v_0$ and $p(\bt\stackrel{g}{\to} \bt')=\gamma^{\tau}_{v_0'v_0}$, the shortest path of \S\ref{gammapar}.
We can now compose with $\rep$ and obtain $\urep:=p\circ\rep$ on objects and morphisms. 
I.e. for $g:\bt\to \bt'$ we have $\urep(g)=U^{\t}_{v'_0v_0}:\H_{v_0}\to \H_{v_0'}$.
 This is not a representation of $\Gpd$, but for $g$ as above and $h:\bt'\to \bt''$ it
satisfies 
\begin{equation}
\label{cocycdefeq}
\urep(h)\urep(g)=\urep(hg)C^-(h,g), \text{ with }  C^-(h,g):=U^{\tau}_{v_0v''_0}U^{\tau'}_{v''_0v'_0}U^{\tau}_{v'_0v_0}:
\in \A_{v_0}
\end{equation}

For the $M$ part of $\rho$ the relevant cocycle will actually be the inverse of $C^-$, see also \S\ref{oppar} below. Explicitly,
$C(h,g)=C^-(h,g)^{-1}=U^{\tau}_{v_0v'_0}U^{\tau'}_{v'_0v''_0}U^{\tau}_{v''_0v_0}$. For three composable morphisms $\bt\stackrel{g}{\to}\bt'\stackrel{h}{\to}\bt''\stackrel{k}{\to} \bt'''$ one obtains 
the following  equation for $C$ by plugging in:
\begin{equation}
\label{Ccocyceq}
C(h,g)C(k,hg)=\nu(g)^{-1}\nu(h)^{-1}\nu(k)^{-1}\nu(khg)=:C(k,h,g)
\end{equation}
And if we denote conjugation of $x$ by $y$ with an upper left index ${}^yx=yxy^{-1}$ to keep with standard notation \cite{Dedecker,Kurosh},
we find the  cocycle equation
\begin{equation}
\label{cocycleeq}
{}^{\nu(g)^{-1}}C(k,h)C(kh,g)= C(h,g)C(k,hg)
\end{equation}

One can also lift the cocycle $C$ to a cocycle $l$ with values in $\Aut$.
\begin{equation}
l(h,g)=\gamma^{\tau}_{v_0v'_0}\gamma^{\tau'}_{v'_0v''_0}\gamma^{\tau}_{v''_0v_0}\in \pi_1(\BG,v_0)
\end{equation}
it satisfies the analogous equation to (\ref{Ccocyceq}) with $\nu$ replaced by $p$. We have $p(l(h,g))=C(h,g)$.

\subsubsection{Matrix version and cocycle}
\label{cyclesec}
In order to do calculations it is preferable to work with a matrix representation of the groupoid action.
The problem is that although the groupoid associates a matrix to each re--gauging, these matrices
all act in different spaces. To make everything coherent one has to use pull--backs. 
Explicitly, for $\bt\stackrel{g}{\to} \bt'$ we set $\mrep(g):=M_g:=M\in \A_{v_0}$ of Proposition \ref{regaugethm}.
If we have another regauging $\bt'\stackrel{h}{\to} \bt''$ then we cannot directly multiply the matrices $M_g$ and $M_h$ as
they have  coefficients in different algebras. We therefore define the product 
$M_h\circ_{\bt}M_g:=U^{\t}_{v_0v_0'}M_h U^{\t}_{v'_0v_0}M_g$.\footnote{Notice that here $U^{\t}_{v_0'v_0}$ is taken to be a ``scalar'' that is it acts as the $k\times k$
diagonal matrix $diag(U^{\t}_{v_0'v_0},\dots,    U^{\t}_{v_0'v_0}):\H_{v_0}^k\to \H_{v_0'}^k$.}
 A straightforward calculation shows:
\begin{prop}
\label{mainprop}
$M_h\circ_{\bt}M_g=C(h,g)M_{hg}$ with the same cocycle $C(h,g)$ as above. \qed 
\end{prop} 

Again, by a straightforward calculation:
\begin{lem}
\label{commlem}
If $\A$ is commutative, then the product is independent of the choice of pull--back $U^{\t}_{v_0v}$. Defining the product by using  conjugation by any $U(\gamma)$ with $\gamma$
a contractible path from $v_0$ to $v_0'$ will give the same result.\qed
\end{lem}
 \begin{cor}
 \label{commcor}
 If the situation is fully commutative, we can use the $\alpha(*v)$ to pull back all the matrices $M_g$  to matrices with coefficients in $\A$. Then the  multiplication above simply becomes matrix multiplication  in $M_k(\A)$.\qed
 \end{cor}
 \subsubsection{Groupoid cocyles and extension}
 \label{oppar}
 The data of $\urep$ and $C$ as well as $p$ and $l$ technically yield a crossed noncommutative groupoid 2-cocycle \cite{Giraud,Dedecker,Kurosh}.
 In order to get one of the standard forms of the cocycle, e.g.\ that of \cite{Dedecker}, we will have to transform the pairs $(p,l)$ and $(\nu,C)$ a bit.
 It turns out that everything is more natural in the opposite groupoid $\G^{op}$ of the groupoid $\G$. This is because we 
 are actually {\em re}--gauging. On the groupoid level define  $p^{op}:\Gpd^{op}\to\P_{\BG}$ and $l^{op}:\Gpd_1^{op}{}_{t}\times_{s}\Gpd^{op}_1\to \Aut$
 \begin{equation}
 p^{op}(g^{op}):=\gamma^{\tau}_{v_0v_0'}, \quad 
  l^{op}(g^{op},h^{op}):=l(h,g)
 \end{equation}
 
 And similarly hitting the above maps with $\rep$ we get
    $\urep^{op}:\Gpd^{op}\to\Btot$ and $C^{op}:\Gpd_1^{op}{}_{t}\times_{s}\Gpd^{op}_1\to \Atot$ 
    \begin{equation}
  \urep^{op}(g^{op}):=\urep(g)^{-1}, \quad
  C^{op}(g^{op},h^{op}):=C(h,g)
 \end{equation}
 
 Now $\Aut$ is a $\P_{\BG}$ crossed module via the inclusion $i:\Aut\to \P_{\BG}$ and the conjugation action $\Phi: \P_{\BG}\times \Aut \to \P_{\BG}: (\gamma,l)\mapsto
 \gamma l\gamma^{-1}$. Analogously $\Atot$ is a $\Btot$ crossed module via inclusion and conjugation action. 
   \begin{prop}
   The pair $(p^{op},l^{op})$ are an element of $C^2_{\P_{\BG}}(\Gpd ,\Aut)$ that is a $\P_{\BG}$--crossed $\Gpd$ 2--cocycle with values in $\Aut$. Likewise the pair $(\urep^{op},C^{op})$ are an element of $C^2_{\Btot}(\Gpd^{op} ,\Atot)$.
   \end{prop}
   
   {\sc Groupoid extension:} By general theory, \cite{Dedecker,Giraud,Moerdijk} the noncommutative cocycle $(p,l)$ gives rise to a groupoid  extension $(\Sigma,b)$ over $\P_{\BG}$
   \begin{equation}
\Sigma:   1\to \Aut\to \hat \Gpd\to \Gpd\to 1 \quad b:\hat \Gpd \to \P_{\BG}
   \end{equation}

   \begin{rmk}
  It is this extension  via $\mrep$ that gives rise to the projective representation  of $\group$ in the commutative case.
  In the noncommutative case, the geometry begins to look like a gerbe geometry. This fits with the non--commutativity being given by
  a 2--form $B$--field. We leave this for further study.
   \end{rmk}
   
  \subsubsection{Categorical description of the Matrix Hamiltonians and re--gauging} 
  The constructions we have presented have a more high--brow explanation. Each spanning tree $\t$ gives
  a functor from $F_{\t}:\P{\BG}\to \P_{\BG/{\t}}$, where $\BG/\t$ is result of contracting $\t$ and is the graph with one vertex $\bullet$ and $k$ loops $l_i$.
Given a functor $(\H,U)$ we can look at all the left Kan--extensions $Lan_{\t}(\H,U):\P_{\BG/\t}\to Hilb$, where $Hilb$ is the category of separable Hilbert spaces. In Particular $\H$ becomes just $\H(\bullet)$.
The action of the $l_i$ is then the action of the diagonal of $\Atot$ obtained by pushing forward $\A_{v_0}$ from a common vertex $v_0$.
In the fully commutative case this coincides with the action of $\A$.

The re--gauging groupoid  compares all the functors obtained by the different left Kan extensions. 

\subsubsection{Re--gaugings Induced by Graph Symmetries}
\label{extendedsympar}
If we have a symmetry, aka.~automorphism, of the graph $\BG$,
then given a fixed choice $\bt$, we can push forward both
these pieces of data to $\bt'$ with $\phi$. This means
that for every $H_{\bt}$ any automorphism $\phi$ gives rise to a re--gauging
of $H_{\bt}$ to $H_{\bt'}$.  That is we have a map
$\group\times \Gpd_1\to \Gpd_1$ where $\Gpd_1$ are the morphisms in $\Gpd$. 

\subsubsection{Lifts to Automorphisms}
One interesting question for any given re--gauging is if there are automorphisms $\psi$ of $\A$ such that
\begin{equation}
\hat \psi(H_{\bt})=U^{\tau}_{v_0v'_0}H_{\bt'}U^{\tau}_{v'_0v_0}
\end{equation}
where, again, $\hat\psi$ is $\psi$ applied to the entries. 
This is the type of enhanced, extended symmetry we will use in the commutative case.

One way such a symmetry can arise is by a re--gauging
induced by an automorphism $\phi$ is $\BG$.
A stricter requirement that is easier to handle is
that not only the matrix coefficients
of the Hamiltonian transform into each other, but rather  already
the weight functions.
This avoids 
dealing with sums of weights. We say a re--gauging induced by an automorphism 
$\phi$ of $\BG$ is 
{\em weight liftable by an automorphism} $\psi$ of $\A$ if 
$\psi(\wt(\vec{e}))=\wt'(\phi(\vec{e}))$, where $\wt'$  is the re--gauged
weight function for the pushed forward spanning tree. 

\begin{thm}
\label{autothm}
Given an automorphism $\phi$ of $\BG$ there is at most one weight lift by an automorphism $\psi$ of the re--gauging induced by $\phi$.
On the generators $\wt(\vec{e})$, $e$ not a spanning tree edge, the putative map is fixed 
by the condition $\psi(\wt(\vec{e})):=\wt'(\phi(\vec{e}))$, where $\wt'$ is the re--gauged
weight function.

Furthermore, the  $\psi(\wt(\vec{e}))$ again generate $\A$ and hence whether  $\psi$ indeed defines  an automorphism
only needs to be checked on the generators $\wt(\vec{e})$.

Lastly, $\psi$ is induced by a base change of $\pi_1(\BG)$.
\end{thm}

\begin{proof}
Let $\wt'$ be the re--gauged weights after moving from $\bt$ to $\bt'$.
If an automorphism $\psi$ of $\A$ that lifts $\phi$ exists, 
then it satisfies $\wt'(\phi(\vec{e}))=\psi(\wt(\vec{e}))$. After
fixing an orientation for each edge the $\wt(\vec{e})$ generate, we see that
the morphism is already fixed, since by assumptions the $\wt(\vec{e})$ generate $\A$.

In order to show that the $\psi(\wt(\vec{e}))$ are generators, we will prove the last statement first.
As discussed in \S\ref{pipar} $\wt$ gives a representation $\rho$ of $\pi_1(\BG,v_0)$ and $\wt'$ gives a representation $\rho'$ of $\pi_1(\BG,v_0')$
if $v_0$ is the root of $\tau$ and $v'_0=\phi(v_0)$ is the root of $\tau'$, the pushed forward spanning tree.
In $\tau$ there is a canonical shortest path $p^{\tau}_{v_0v'_0}$ from $v'_0$ to $v_0$. Conjugating by this path
gives an isomorphism $P:\pi_1(\BG,v_0')\to \pi_1(\BG,v_0)$. This is in essence the definition of the path groupoid of $\BG$.
Let $l^{\tau}(\vec{e})$ be the loop associated to $\vec{e}$ by using $\tau$ as a spanning tree, see \S\ref{pipar},
then $\wt(\vec{e})=\rho(l^{\tau}(\vec{e}))$. It follows from the definition of the re--gauging that 
$$\psi(\rho(l^{\tau}(\vec{e})))=\psi(\wt(\vec{e}))=wt'(\phi(\vec{e}))=\rho'(l^{\tau'}(\phi(\vec{e})))=\rho(P(l^{\tau'}(\phi(\vec{e}))))$$
so that  $\psi$ is induced by the chance of basis $l^{\tau}(\vec{e})\to P(l^{\tau'}(\phi(\vec{e})))$ in $\pi_1(\BG,v_0)$.
From this it follows that the  $\psi(\wt(\vec{e}))$ generate.
\end{proof}

\begin{cor}
\label{nondegcor}
If the groupoid representation is non--degenerate, so that 
$\A$ is generated by the $\wt(\vec{e})$ and each non--spanning--tree edge gives a linearly
independent generator, then the morphism $\Psi$ above is well--defined as a linear morphism.

If there are no relations among the generators, e.g.\ in the case $\A=\T^n$ the commutative algebra of the torus,
then every automorphism is weight liftable, i.e.\ $\Psi$ from above is well--defined as an algebra homomorphism.
\qed
\end{cor}

We will use the corollary in \S\ref{compar} to define the enhanced symmetry groups in the commutative toric non--degenerate case.

\begin{rmk}
\label{symrmk}
These types of symmetries might also help to explain the somewhat mysterious approximate symmetry between the noncommutative
and the commutative case found in \cite{kkwk2}. Here the symmetry is between two loci in the base of the variations. 
In the commutative case this is the base $X$ for the family of
Hamiltonians as discussed above and in the noncommutative case is the space parameterizing the background magnetic field $B$.
The two loci are the locus of degenerate Eigenvalues on the commutative side and the locus of values of $B$ where $\B$ is not
the full matrix algebra. From the examples PDG and honeycomb, these two loci have exactly the same top dimension and
there are further characteristic features which they have in common.
These considerations would lead us too far astray in the present context, but we plan to return to them in a future paper.

\end{rmk}

\subsection{Enhanced symmetries in the commutative case}
\label{compar}

We will concentrate on the commutative case in the following.  One physical feature that makes the noncommutative theory more complicated is that  conjugating $H$ with elements from $\A$ usually does not leave it invariant. This is of course the starting point for considering  the $C^*$--algebra $\B$ which
contains all these conjugates.
\subsubsection{Extension}
\label{extensionpar}
In the commutative case, $U(\Aut)$ is a commutative group and the 2--cocycle
defines a central extension $\tilde \Gpd$ of $\Gpd$ by $U(\Aut)$. We can consider the action of this central extension, since
the action of $U(\Aut)$ commutes with the Hamiltonians, permutations and the re-gaugings in this case.
If we are moreover in the fully commutative case, then by using the diagonal embedding of $\A$ we can even
make the cocycle take values in $U(\A)$ and hence obtain the central extension.

\begin{equation}
1\to U(\A)\to \tilde\Gpd \to \Gpd\to 1
\end{equation}
Then $\rho$ does give a groupoid representation of $\tilde \Gpd$.

\begin{rmk}
There is a nice geometric interpretation of this in the case of wire networks. Here the group $U(\A)$ corresponds
to translations along the lattice $L$. One can identify the vertices of $\BG$ with the elements in a chosen
primitive cell and likewise one can arrange the spanning tree edges to be inside this cell. When we are re--gauging, we move the base point along the spanning tree edges. After doing this several times the new root can lie outside the original 
primitive cell. The co--cycle then measures the displacement of the new cell relative to the old cell in terms of 
an element $\lambda\in L$, more precisely it is just $U_{\lambda}$. 
\end{rmk}

\subsubsection{Enhanced symmetry group}

In order to find degeneracies in the spectrum, we use the characters and then look for fixed points
under the induced groupoid action.
Using the language of \S\ref{comgeopar}, given a point $\chi \in \A$ we get a map $\hat\chi:Ham_0\to M_k(\C)$. There is then an induced
action of the groupoid on $\hat\chi(Ham_0)$, by pushing forward with this character.
It can now happen that  $\hat\chi(H_{\bt})=\hat\chi(H_{\bt'})$, that is $H_{\bt}(t)=H_{\bt'}(t)$, for the point $t\in X$ corresponding to $\chi$.

For each element
$H(t)\in \hat\chi(Ham_0)$, we get its stabilizer group $\stab(H(t))$ under the induced groupoid action.
This is the image of the transitive action of the groupoid on the fiber of $\hat \chi$ over $H(t)$.
We can identify $\stab(H(t))$ with the image of that subgroupoid.
If this group is not trivial, which means that the fiber is not just a point, we call this
group the enhanced symmetry group of $H(t)$. It is realized by re--gaugings, that is conjugation by
specific matrices which form a projective representation of the stabilizer group as we presently discuss.

\subsubsection{Super--selection rules, Projective Representation and Degeneracies}

If $\stab(H)\neq 1$ then this means that 
 the set of all matrices $\hat \chi\rho(g)$
for $g\in \stab(H(t))$, where we 
identified $g$ with its defining element in $\Gpd$,
 all commute with the Hamiltonian $H(t)$ and
 hence each one and all of them together give super--selection rules.
This of course is already a great help in finding the spectrum.

Since $\rho$ is only a groupoid representation of $\tilde \Gpd$, 
we get that $\hat\chi\rho$
is a  representation of a an extension of  $\stab(H(t))$.
 If we are in the fully commutative case
this extension is central and gives rise to a 
{\em projective} representation of $\stab(H(t))$.
\begin{equation}
1\to U(1) \to \tilde \stab(H(t))\to \stab(H(t))\to 1
\end{equation}
Here we pulled back with the diagonal embedding, i.~e. $U(1)$ is embedded as scalars, viz.\ diagonal matrices.

In order to apply the general arguments of representation theory, we will be interested in the class of this extension. These extensions are classified
up to isomorphism by $H^2(\stab(H(t)),U(1))$ \cite{Schur10,kap}.

We give a brief definition of this cohomology group,
as it is important
for our calculations (see e.g.~\cite{brown}).
Let $G$ be a group and $A$ be an Abelian group, which we also write multiplicatively\footnote{We  consider $A$ to be a $G$--module with the trivial action.}.
Set $C^i(G,A):=Map(G^{\times i},U(1))$ these are the $i$--th cochains.
There is a general differential $d:C^i\to C^{i+1}$ with $d^2=0$. We will need the formulas for it on 1-- and 2-- cochains. If $\lambda \in C^1(G,A)$ then $d\lambda(g,h)=\lambda(g)\lambda(h)\lambda(gh)^{-1}$
and if $c \in C^2(G,A)$ then $dc(g,h,k)=c(h,k)c^{-1}(gh,k)c(g,hk)c(g,h)^{-1}$.
Set $Z^2(G,A):=ker(d:C^2(G,A)\to C^3(G,A))$ and $B^2(G,A):= Im(d:C^1(G,A)\to C^2(G,A))$.
Notice that an element $c\in C^2$ is in $Z^2$ precisely means that $c$ satisfies the cocycle condition (\ref{cocycleeq}) in the Abelian case, where the conjugation action is trivial. Now $B^2(G,A)\subset Z^2(G,A)$ and $H^2(G,A):=Z^2(G,A)/B^2(G,A)$.

What this means is that we can move to an isomorphic extension by using a rescaling $\lambda\in C^1(\stab(H(t)),U(1))$. Another interesting concrete question is if a given homology class $[c]$ can
be represented by a cocycle in a subgroup of $A$. This is especially interesting if the subgroup is finite. In our concrete calculations for the Gyroid, we will use for instance
$\Z/2\Z$ and this will lead us to consider double covers.

In general if we identify that the projective action of $\stab(H(t))$ is isomorphic to an action
of a finite group extension $\tilde \stab(H(t))$, then we can use this representation to
decompose $\C^k$ into its isotypical decompositions with respect to this group action.
If the group is non--Abelian, then there is a chance that some of the irreducible representations
in the decomposition are higher dimensional,  which implies degeneracies of the order of these dimensions. Again this is present for the Gyroid.

\subsubsection{Geometric lift of the groupoid action}
In order to understand the (projective) group action,
geometrically in the commutative case,
 one lifts the action on the Hamiltonians
to an action on the underlying geometric space. We will now for concreteness fix $\A=\T^n$,
that is the groupoid representation is commutative, toric non--degenerate,
as is the case in all crystal examples we  consider: PDG, Bravais and Honeycomb.

Finding lifts then means that one considers the commutative diagram

\begin{equation}
\label{actioneq}
\xymatrix{
T^n\ar[rr]^{H_{\bt}} \ar@{.>}[d]_{\Psi_{\bt'}^{\bt}}
\ar[rrd]^{ H_{\bt'}}
&&\hat\chi (Ham_0)\ar[d]^{\Phi_{\bt'}^{\bt}}\\
T^n\ar[rr]_{ H_{\bt}}& &\hat\chi(Ham_0)
}
\end{equation}
where the dotted morphism is the lift to be constructed
and $H_{\bt}$ is the map $t\mapsto H_{\bt}(t):=\hat\chi(t)$ if $\chi$ is the character
corresponding to $t$.

The existence of these lifts is not guaranteed in general, and indeed there are examples
of re--gaugings that cannot be lifted. A non--liftable example can e.g.\ be produced from
the cube graph obtained from the Gyroid graph by quotienting out by the the simple cubic lattice,
see \cite{kkwk}. We will show that all lifts stemming from automorphisms of the underlying graph do lift.

Looking
at the diagram (\ref{actioneq}),
one consequence of this action is that it lets us pinpoint Hamiltonians with enhanced symmetry group.
Using Corollary \ref{nondegcor} and translating it to the geometric side, we obtain

\begin{prop}
\label{nondegprop}
Let $(\bar \Gamma,\wt)$ be a toric non--degenerate weighted graph. 
In the commutative case the automorphism group of $\BG$ lifts via 
the gauging action to an automorphism group of $\T^n$. 
That is we get a morphism $Aut(\BG)\to Aut(\T^n)$.

 If a point $t\in T^n$ is a fixed point of a lift of an element $g\in \Gpd$, then,
the re--gauging is an enhanced symmetry for the corresponding Hamiltonian,
 that is $\hat\chi(\rho(g))$ commutes with the Hamiltonian $H_{\bt}(t)$.
\qed
\end{prop}

Summarizing these results:

\begin{thm}
\label{nondegthm}
If $(\BG,\wt)$ is commutative and  toric non--degenerate, then  
a  stabilizer sub--group $G_t$ of $t\in T^n$ 
under the induced action
of $Aut(\BG)$ on $T^n$
leads to an
enhanced symmetry group $\stab(H(t))$ for the Hamiltonian $H(t)$.
This group also has a projective 
representation via the matrices $\hat\chi(\rho(g))$.
\qed
\end{thm}

We can exploit the representation theory of this group to get 
information about degeneracies.

\begin{ex}
\label{symex}
For commutative toric non--degenerate groupoid representations of
 symmetric graphs the re--gaugings by re--orderings are always representable
via an automorphism of the graph.
 If $\sigma\in \SS_{k-1}$ permutes
the vertices of the spanning tree leaving the root fixed, then the re--gauging
lifts as the reordering of the generators and possibly taking $*$ of them. 
The matrices are just the usual permutation matrices
of $\SS_{k-1}\subset \SS_{k}$ acting on the last $k-1$ copies of $\C$ in $\C^k$.
\end{ex}

\begin{rmk}
In the commutative case, the representation of $\pi_1(\BG,v_0)$ factors through its Abelianization  $H_1(\BG)$.
\end{rmk}

\section{Calculations and Results for Wire Networks}
\label{calcpar}
In this paragraph, we perform the calculation for the PDG and honeycomb graphs of Figure \ref{graphsfig}.
These correspond to wire networks as reviewed in 
\S\ref{wiregeosec}. In all these situations 
Theorem \ref{nondegthm} applies.
The upshot of the following calculations together 
with the analysis of \cite{kkwkdirac} is:

\begin{thm}
In all the examples  PDG and honeycomb, all the fixed points come from fixed points of the semi--classical action. Moreover the fixed points, stabilizer groups, their extensions
and the decomposition into irreps for the case of the Gyroid are given in Table 
\ref{gyroidtable}.
\end{thm}

For the calculations, we note that we are in the fully commutative case and hence Corollary  \ref{commcor} applies.
Furthermore the graphs  all have transitive symmetry groups, so that we only have to calculate for one source $\bt$
and can then transport the results by push--forward to any other.
\subsection{Gyroid}

The graph $\BG$ in the Gyroid case is the full square. It has symmetry group $\SS_4$.
With the Gyroid weights, the graph is faithful and hence the action
can be lifted to an action on the torus.
 It acts transitively on all ordered spanning trees. Such a spanning tree is fixed by specifying a root and the order.
 The subgroup of $\SS_3$ acts transitively
 on all orders. The matrices of this subgroup action are just the permutation matrices acting
 on the last $3$ copies of $\C$ in $\C^4$ and
 the lift of the $\SS_3$ action on the generators $A,B,C$ of $\T^3$ is given by the permutation action.

We fix an initial rooted spanning tree and order as in Figure (\ref{t12fig}).

\subsubsection{Action on $T^n$}
The action of $\SS_4$ on $T^3$ is fixed once we know the action of the generators $(12),(23)$ and $(34)$.

The action of $(23)$ is graphically calculated
in Figure \ref{t23fig}, from which one reads off
$\Psi((23))(A,B,C)=(A^*,C^*,B^*)$. Here $(A,B,C)$ is the notation for the initially chosen basis of $\T^3$.

In the graphical calculation, we first write down the graph together with the
initial spanning tree and order. We then push--forward the spanning tree and 
the order. For this we keep the vertices and edges as well as the weights fixed.
We then (if necessary) give the re--gauging parameters by writing
them next to the respective vertices and (if necessary) perform the re-gauging.
Finally we move the vertices and edges, so that they coincide with their pre--images to read off the morphism on the generators given by Theorems \ref{autothm}
and \ref{nondegthm}.

\begin{figure}
\includegraphics[width=.9\textwidth]{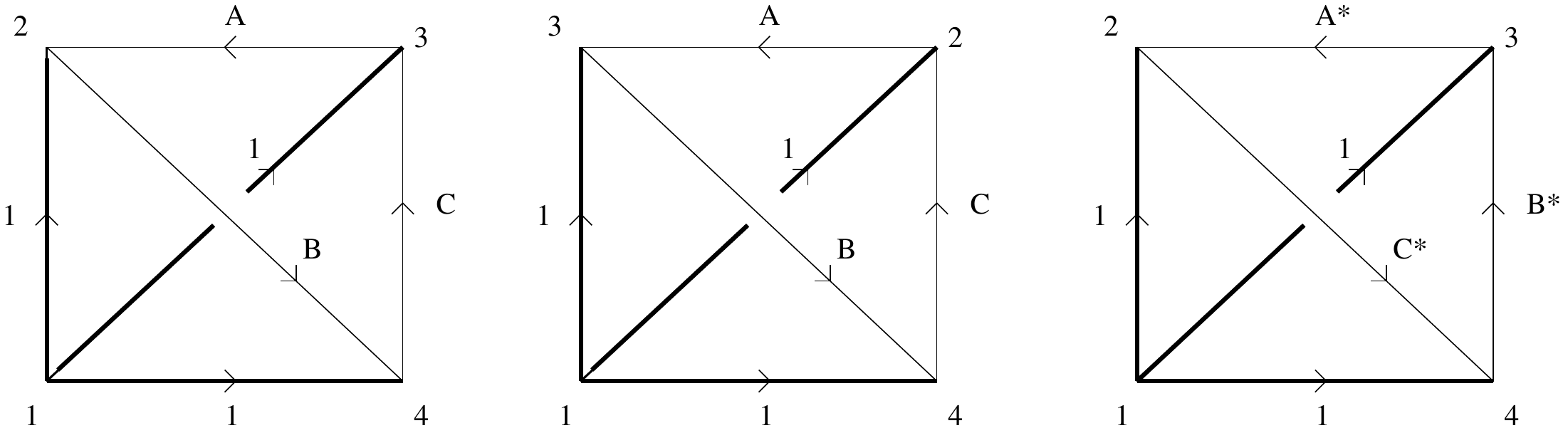}
\caption{Calculation of the action of $(23)$ on $T^3$. The original graph,
the pushed forward order and the move into the old position to read off the morphism}
\label{t23fig}
\end{figure}
A similar calculation shows that $\Psi((34))(A,B,C)=(B^*,A^*,C^*)$.
A consequence is that the cycle $(234)=(23)(34)$ acts as
$\Psi((234))(A,B,C)=\Psi((23))(B^*,A^*,C^*)=(B,C,A)$ and is the cyclic permutation.

The action of $(12)$ is more complicated as the root is moved.
For this we calculate graphically, see Figure \ref{t12fig}, and read off $\Psi$ as:
$(A,B,C)\mapsto (A^*,B^*,ACB)$.

\begin{figure}
\includegraphics[width=\textwidth]{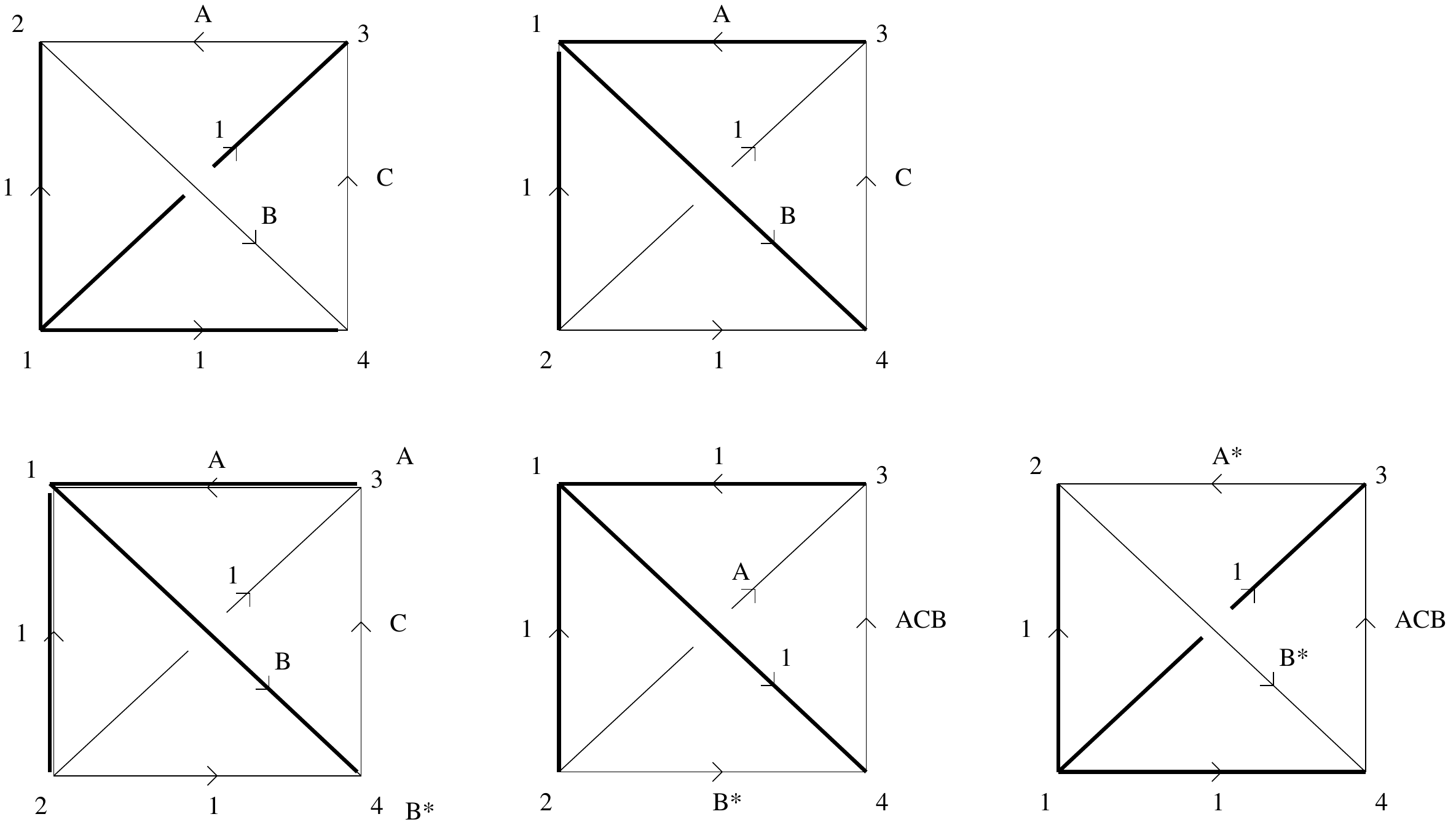}
\caption{Calculation of the action of $(12)$ on $T^3$}
\label{t12fig}
\end{figure}

This allows us to compute fixed points and stabilizer groups.
We will first concentrate on non--Abelian stabilizer groups.
There are only two fixed points under the full $\SS_4$ action and these
are $(1,1,1)$ and $(-1,-1,-1)$.
The group $A_4$, the subgroup of
all even permutations, is the stabilizer
group of the two points $(i,i,i)$ and $(-i,-i,-i)$.
One can readily check that these are the only non--Abelian stabilizer groups. The other possibility
would be $\SS_3$, but a short calculation shows that anything that is stabilized by any $\SS_3$ subgroup is stabilized by all of $\SS_4$.

\subsubsection{Representations}
We collect together the matrices $M$ needed for further calculation. Again,
we fix our initial ordered rooted spanning tree as before.

Using short hand notation,
the matrix for the re--gauging induced by the transpositions $(12),(13),(14)$ from the initial spanning tree to the pushed forward
 one are
 $$\rho_{12}=\left(\begin{matrix}0&1&&\\1&0&&\\&&A&\\&&&B^*\end{matrix}\right), \rho_{13}=\left(\begin{matrix}0&&1&\\&A^*&&\\1&&0&\\&&&C\end{matrix}\right),\rho_{14}=\left(\begin{matrix}0&&&1\\&B&0&\\&0&C^*&\\1&&&0\end{matrix}\right)$$
The calculation for $\rho_{12}$ can be read off from Figure \ref{t12fig}.
For this we read off the matrix $\Phi$ from the re--gauging parameter and
the matrix $M_{\sigma}$ is given by the permutation we are considering.
The other calculations are similar.
All other transpositions, viz.\ those not involving $1$, 
simply yield permutation matrices as there is no re--gauging involved.
It is convenient to also have the following matrices as a reference:
$$\rho_{(12)(34)}=\left(\begin{matrix}0&1&&\\1&0&&\\&&0&A\\&&B^*&0\end{matrix}\right), \quad \rho_{(14)(23)}=\left(\begin{matrix}&&&1\\&&B&\\&C^*&&\\1&&&\end{matrix}\right)
$$
and finally
$$
\rho_{(123)}=\left(\begin{matrix}0&0&1&\\1&0&0&\\0&A&0&\\&&&B^*\end{matrix}\right)
$$
\subsubsection{The point $(0,0,0)$}
At $(0,0,0)$, the matrices $\rho_{12},\rho_{23},\rho_{34}$ give the usual representation of $\SS_4$ on $\C^4$. As is well known this representation decomposes  into the trivial representation and an irreducible 3--dim representation. This means that there is an at least 3--fold degenerate Eigenvalue $\lambda$. Since the trace of $H$ is zero, we also know that the Eigenvalues satisfy $\mu=-3\lambda$. Plugging in $(1,1,1,1)$, which spans the trivial representation, we see that $\mu=3$ and $\lambda=-1$.

\subsubsection{The point $(\pi,\pi,\pi)$}
In this case, the matrices $\rho_{12},\rho_{23},\rho_{34}$ only give a projective representation.
As one can check  $\rho_{12}\rho_{23}\rho_{12}=-\rho_{13}$ while
$\rho_{23}\rho_{12}\rho_{23}=\rho_{13}$ for instance.
Define the 1--cocycle $\lambda$ by $\lambda(\sigma)=(-1)$ if $1$ appears in a cycle of length $>1$
and $1$ else.
So that $\lambda((12))=\lambda((13))=\lambda((123))=-1$ while $\lambda((23))=\lambda((24))=\lambda((234))=1$.
Then one calculates that $\tilde \rho:=\rho\circ \lambda$ has a trivial cocycle $c$ and thus
$\rho$ is isomorphic to a true linear representation of $\SS_4$. Checking
the characters, one sees again that in this case the irreducible
components of $\tilde\rho$, which also commute with $H$ are again the one--dimensional trivial
representation and the 3--dimensional standard representation. The trivial representation is spanned by $(-1,1,1,1)$.
 The Eigenvalues are then readily computed to be
$1$ with multiplicity $3$ and $-3$ with multiplicity $1$.

\begin{rmk}
We would like to remark that the choice of $\lambda$ amounts 
to choosing a different gauge
for the root vertex, namely $-1$ instead of $1$.
\end{rmk}

\begin{rmk}
Notice that already in this case, even though there is no projective extension, 
our enhanced gauge--group is necessary. 
Without it there would only be an $\SS_3$ action, those elements
which involve no re--gauging. This smaller symmetry group is, however, not powerful
enough to force the triple degeneracy, as 
$\SS_3$ has no irreducible $3$--dim representation.
\end{rmk}

\subsubsection{The point $(\pi/2,\pi/2,\pi/2)$ and $(-\pi/2,-\pi/2,-\pi/2)$}
These points are similar to each other. We will treat the first one in detail.
Again, we have only a projective representation of $A_4$ aka.\ the tetrahedral group $T$.
Namely, $\rho_{(12)(23)}\rho_{(13)(24)}=-i\rho_{(14)(23)}$. Again we can scale
by a 1--cocycle $\lambda$. This time $\lambda(id)=1$,
$\lambda((ij)(kl))=i$, $\lambda(ijk)=1$ if $1\notin \{i,j,k\}$,
and $\lambda((ijk))=i$ if $1\in \{i,j,k\}$.
The resulting representation $\tilde\rho=\rho\circ\lambda$ is then still a projective
representation, but is it a representation of the unique non--trivial $\Z/2\Z$ extension of
$A_4$, which goes by the names $2T,2A_4,SL(2,3)$ or the binary tetrahedral group.
This group is well known. It is presented by generators $s$ and $t$ with the relations $s^3=t^3=(st)^2$.
In $SL(2,3)$ (that is the special linear group of $2\times2$ matrices over the field with three elements ${\mathbb F}_3$), one can choose $s=\left(\begin{matrix}-1&-1\\0&-1\end{matrix}\right)$
and $t=\left(\begin{matrix}-1&0\\-1&-1\end{matrix}\right)$. 

For $2A_4$ using a set theoretic section $\wedge$
of the extension sequence
\begin{equation}
\xymatrix{
1\ar[r]&\Z/2\Z\ar[r]&2A_4\ar@<1ex>[r]&A_4\ar@<1ex>[l]^{\wedge}\ar[r]&1
}
\end{equation}
and $z$ as a generator for $\Z/2\Z$,  we can 
pick $s=z\widehat{(123)},t=z\widehat{(234)}$ 
as generators. 
Now we can check the character table, Table
\ref{chartable},  and find that the representation $\tilde \rho$
over the complex numbers decomposes as the sum of two irreducible two--dimensional representations $\chi_5\oplus \chi_6$.
In fact, these are the two representations into which the 
unique real irreducible 4--dimensional representation of complex type
splits over $\C$.

The explicit computation for the representation 
\begin{eqnarray}
\tilde\rho(s)&=&-\lambda((123))\rho_{(123)}=
\left(\begin{matrix}0&0&-i&0\\-i&0&0&0\\0&1&0&0\\0&0&0&-1\end{matrix}\right)\nn\\
\tilde\rho(t)&=& -  \lambda((234))\rho_{(234)}=
-\left(\begin{matrix}
1&0&0&0\\0&0&0&1\\0&1&0&0\\0&0&1&0
\end{matrix}
\right)
\end{eqnarray}
is as follows. Suppose the $\tilde\rho=\bigoplus_{i=1}^7a_i\rho_i$, where
$\rho_i$ is the irrep with character $\chi_i$. 
Now $tr(id)=4,tr(-1)=-4$ , using 
the character table this implies that the coefficients $a_1=a_2=a_3=a_7=0$
and furthermore $(*)\,  a_4+a_5+a_6=2$. We furthermore have that $tr(s)=-1$
so that $a_4+\omega a_5 +\omega^2 a_6=-1$ which together with (*) implies
that $a_4=0, a_5=a_6=1$. This fixes the decomposition into irreps.
As a double check one can verify 
that the rest of the equations are also satisfied.

So indeed we find that $(\pi/2,\pi/2,\pi/2)$ is a point with two Eigenvalues with degeneracy $2$.
It is not hard to find (e.g.\ using the results of \S\ref{superpar}) that these Eigenvalues
are $\pm\sqrt{3}$.

\begin{table}

\begin{center}
\begin{tabular}{|c|c|c|c|c|c|c|c|}
\hline
Representative&$1$&$-1$&$s^3$&$t^2$&$s^2$&$t$&$s$\\
\hline
Elts in Conj. Class & $1$ & $1$ & $6$ & $4$ & $4$ & $4$ & $4$ \\
\hline
Order & $1$ & $2$ & $3$ & $3$ & $4$ & $4$ & $6$ \\
\hline
\hline
$\chi_1$ & $1$& $1$ & $1$ & $1$ & $1$ & $1$ & $1$\\
\hline
$\chi_2$ & $1$& $1$ & $1$ &$\omega$ & $\omega^2$ & $\omega^2$ & $\omega$\\
\hline 
$\chi_3$ & $1$& $1$ & $1$& $\omega^2$ & $\omega$ & $\omega$ & $\omega^2$ \\
\hline
$\chi_4$ & $2$& $-2$ &$0$ & $-1$ & $-1$ & $1$ & $1$ \\
\hline
$\chi_5$ & $2$& $-2$ & $0$&$-\omega$ & $-\omega^2$ & $\omega^2$ & $\omega$\\
\hline
$\chi_6$ & $2$& $-2$ &$0$& $-\omega^2$ & $-\omega$ & $\omega$ & $\omega^2$ \\
\hline
$\chi_7$ & $3$& $3$ & $-1$& $0$ & $0$ & $0$ & $0$\\
\hline
\end{tabular}
\caption{\label{chartable} Character table of $2\cdot A_4$\cite{Schur07}, where $\omega = e^{\frac{2 \pi i}{3}}$.}
\end{center}

\end{table}

The analysis of the complex conjugate point $(-\pi/2,-\pi/2,-\pi/2)$ is analogous.

We would briefly like to connect these results to \cite{kkwkdirac}.
There it was shown that these four points are the only singular points in the spectrum
and that the two double crossing points are Dirac points.

 \begin{table}

 \begin{tabular}{llllll}
 $a,b,c$&Group&Iso class of&type&Dim of &Eigenvalues $\lambda$\\
 &&of extension&&Irreps&\\
 \hline
 $(0,0,0)$&$\SS_4$&$\SS_4$&trivial&1,3&$\lambda=-1$ three times\\
 &&&&&$\lambda=3$ once\\
 \hline
 $(\pi,\pi,\pi)$&$\SS_4$& $\SS_4$ &trivializable &1,3&$\lambda=1$ three times\\
 &&&\phantom{$\SS_4$} cocycle&&$\lambda=-3$ once\\
 \hline
 $(\frac{\pi}{2},\frac{\pi}{2},\frac{\pi}{2})$& $A_4$&$2A_4$& isomorphic&2,2&$\lambda=\pm\sqrt{3}$ twice each\\
 $(\frac{3 \pi}{2},\frac{3 \pi}{2},\frac{3 \pi}{2})$&&\phantom{$1A_4$} &extension&\\
 \hline
 \end{tabular}
  \caption{\label{gyroidtable}
 Possible choices of parameters $(a,b,c)$ leading to non--Abelian enhanced symmetry groups and degenerate eigenvalues of $H$}
 \end{table}

\subsubsection{Super--Selection Rules and Spectrum Along the Diagonal}
\label{superpar}
To illustrate the power of the super--selection rules we consider the action of the cyclic group generated by $(234)$. One can easily see that the fixed point set in the $T^3$ is the diagonal $t=(a,a,a)$. The matrices $\hat\chi\rho$ actually give a {\it bona fide} representation of $\C^4$. This is the
representation of $C_3$ given by cyclicly permuting
the last three factors of $\C$.
The action decomposes into irreps as follows: ${\mathbb C}^4=triv\oplus triv\oplus \omega\oplus \bar\omega$. Where $\omega$ is the 1--dimensional representation given by $\rho((123))=\omega=exp(2\pi i /3)$.
The two trivial representations are spanned by $v_1=(1,0,0,0)$ and $v_2=(0,1,1,1)$ while
the representation $\omega$ is spanned by $w=(0,1,\omega,\bar\omega)$ and $\bar\omega$ by $\bar w$.

 Although we cannot extract information about the degeneracies from this it  helps greatly in determining the Eigenvalues, since there are two irreps with multiplicity one each giving a unique 1--dimensional  Eigenspace for
 the Hamiltonian. Hence we immediately get two Eigenvalues.
Plugging $w$ and $\bar w$ into $H(t)$ one reads off
\begin{equation}\lambda_1=\omega \, exp(ia)+\bar{\omega} \exp(-ia)\quad \lambda_2=\bar{\omega} \exp(ia)+\omega \exp(-ia)
\end{equation}

 The sum of the two trivial representations gives a 2--dimensional isotypical component. Therefore, we have to diagonalize $H$ inside this
 Eigenspace. It is interesting to note
that at the special points it is exactly this flexibility that is needed in order to allow for crossings.

To determine the two remaining eigenvalues $\lambda_3$ and $\lambda_4$,
we apply $H$ to $\vec{v}=xv_1+yv_2=(x,y,y,y)$. The eigenvalue equation
$H\vec{v} =\lambda \vec{v}$ leads to the equations
$ 3y=\lambda x$ and  $ x+y (\exp(ia)+\exp(-ia))=\lambda y$,
Fixing $x=3$ this gives the quadratic equation
$\lambda^2-2\cos(a)\lambda-3=0$
 which has the two solutions
 \begin{equation}
 \label{evaleq}
\lambda_{3,4}=\cos(a)\pm \sqrt{\cos^2(a)+3}
\end{equation}

This gives the spectrum along the diagonal which is given in  Figure \ref{dispersion}. The calculation only involves
the classical symmetries without re--gauging.

 \begin{figure}
\includegraphics[width=.6\textwidth]{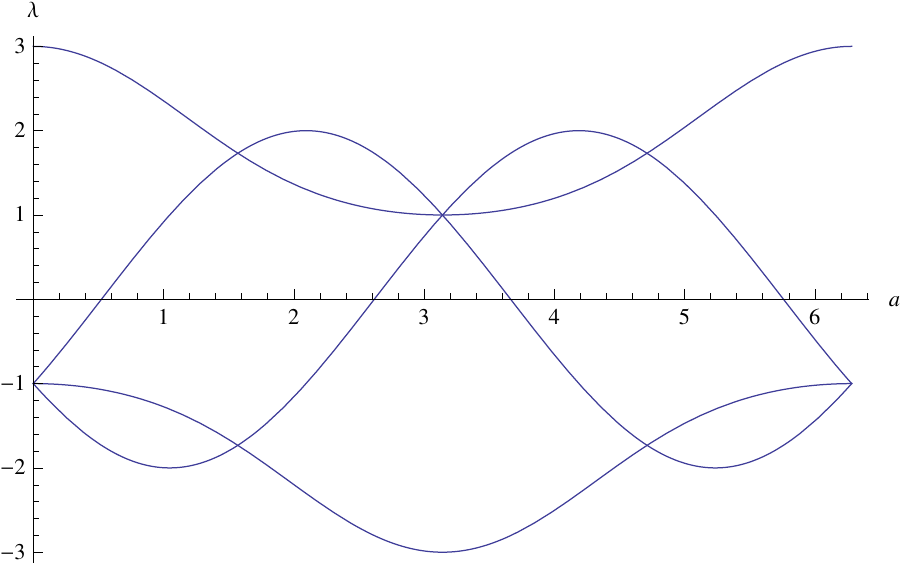}
\caption {Spectrum of $H$ along the diagonal in $T^3$}
 \label{dispersion}
\end{figure}

In \cite{Avron} the authors also assert that numerically they only found
singular points in the spectrum along the diagonal.
The fact that the arising candidates for Dirac points are indeed such points
and the analytic proof that indeed there are no other singular points
in the spectrum is contained in \cite{kkwkdirac}.

\subsection{The P case.} There is nothing much to say here. There is only the root of the spanning tree which is unique. The $\SS_3$ action permutes
the edges and their weights. This yields the permutation action on the $T^3$. 
There is no nontrivial cover and the Eigenvalues remain invariant.

\subsection{The D case.} Here things again become interesting. 
Permuting the two vertices, we obtain eight fixed points 
if $a,b,c\in \{1,-1\}$. The matrix for this transposition is $\left(\begin{matrix}0&1\\1&0\end{matrix}\right)$. This gives super--selection rules and we know
that $v_1=(1,1)$ and $v_2=(-1,1)$ are Eigenvectors. The Eigenvalues being $1+a+b+c$ and $-(1+a+b+c)$ at these eight points.

We can also permute the edges with the $\SS_4$ action. In this case the $\SS_3$ action
leaving the spanning tree edge invariant acts as a permutation on $(a,b,c)$. The relevant
matrices however are just the identity matrices and the representation is trivial.
The transposition $(12)$, however,  results in the action $
(a,b,c)\mapsto (\bar a,\bar a b,\bar ac)$ on $T^3$, see Figure \ref{t12fig}.
So to be invariant we have $a=1$, but this implies that $\rho_{12}$ is the identity matrix.
Invariance for $(13)$ and $(14)$ and the three cycles containing $1$ are  similar. 
But, if we look at invariance under the element $(12)(34)$ we are lead to the equations
$$
a=\bar a, b=\bar a c,c=\bar a b
$$
This has solutions $a=1,b=c$, for these fixed points again we find only a trivial action.
But
for $a=-1,b=-c$ these give rise to the diagonal matrix $diag(1,-1)$ and hence Eigenvectors $e_1=(1,0)$ and $e_2=(0,1)$, but looking at the Hamiltonian,
these are only Eigenvectors if it is the zero matrix $H(a,b,c)=0$. 
Indeed the conditions above imply $1+a+b+c=0$. 

Similarly, we find a $\Z/2\Z$ group for $(13)(24)$
and $(14)(23)$ yielding the symmetric equations $b=-1$, $a=-c$ and $c=-1,a=-b$.
These are exactly the three circles found in \cite{kkwk2}.
Going to bigger subgroups we only get something interesting if the stabilizer group
$G_t$ contains precisely two of the double transposition above. That is the Klein four group $\Z/2\Z\times
\Z/2\Z$. The invariants  are precisely the intersection points of the three circles
given by $a=b=-1$ and $c=1$ and its cyclic permutations.

These are three lines of double degenerate Eigenvalue $0$. These are not Dirac points
since there is one free parameter and hence the fibers of the characteristic
 map of \cite{kkwkdirac} are one dimensional which implies that the singular point is not isolated.

\subsection{The honeycomb case}
This is very similar to the D story. The vertex interchange renders the 
fixed points $a=\pm1,b=\pm1$ which have Eigenvectors $v_1,v_2$ as above and Eigenvalues
 $1+a+b$ and $-(1+a+b)$ respectively. The irreps of the $C_3$ action are $triv\oplus \omega$.

As far as the edge permutations are concerned the interesting one is 
 the cyclic permutation $(123)$ which yields the equations
$$
a=\bar b, b=\bar b a
$$
for fixed points. Hence $a^3=1$. We get non--trivial 
 matrices at the points $(\omega,\bar\omega)$ and $(\bar\omega,\omega)$. At these points $e_1,e_2$ are Eigenvectors with Eigenvalue $0$ and $H=0$,
since $1+a+b=1+\omega+\bar\omega=0$.
They are exactly the Dirac points of graphene.

\section{Conclusion}
By considering re--gaugings, we have found the symmetry groups
fixing the degeneracies of the PDG and honeycomb families of graph
Hamiltonians.
The symmetries we used were those induced by the automorphisms of 
the underlying
graphs. In our specific examples, all the graphs were highly symmetric, and hence had large  automorphism 
group.
Here we stress that our symmetries are extended
symmetries and not just the classical ones. The most instructive and
interesting case is the action of the binary tetrahedral group giving
rise to the Dirac points in the Gyroid network. Note that as dimension--$0$ objects,
the Dirac 
points for  the Gyroid are codimension--3 defects in $T^3$,
rather than codimension-2 defects in
$T^2$ as
for the honeycomb lattice, which describes graphene. 
Nevertheless, one may expect that they too lead to special physical 
properties.

There are several questions and research directions 
that tie into the present analysis.

It would be interesting to find concrete examples of lifts of re--gaugings
either in the noncommutative case or in the case of re--gaugings not
induced by graph symmetries. One place where we intend to look for the former
 is in the noncommutative case of PDG and the honeycomb as we aim to probe
the noncommutative/commutative symmetry mentioned in  Remark \S\ref{symrmk}.

We are furthermore interested in how these symmetries behave under deformations
of the Hamiltonian and if they are topologically stable. A physically important 
type of deformations are those corresponding to periodic (in space) 
lattice distortions that describe crystals with lower spatial symmetry than those
considered here. Such distortions may occur for instance during
synthesis of the structure \cite{Hillhouse}.
Codimension-3 Dirac points, such as those of the Gyroid network,
are especially interesting in this respect: they can
be viewed as magnetic monopoles in the parameter 
space \cite{Berry} and as such are expected 
to be topologically stable. This makes the physics associated with such points immune
to periodic lattice distortions.

Finally it seems that on the horizon there are connections between our
theory and two other worlds. The first being quiver representations in 
general and the second  being cluster algebras. The connection to the
first is inherent in the subject matter, while the connection 
to the second needs some work. The point is that in our transformations,
we change several variables at a time. Nevertheless, the re--gauging groupoid
can be viewed as a sort of mutation diagram. We plan to investigate
these intriguing connections in the future.

\section*{Acknowledgments}
RK thankfully acknowledges
support from NSF DMS-0805881 and DMS-1007846.
BK  thankfully acknowledges support from the  NSF under the grants PHY-0969689 and PHY-1255409.
 Any opinions, findings and conclusions or
recommendations expressed in this
 material are those of the authors and do not necessarily
reflect the views of the National Science Foundation.
Both RK and BK thank the Simons Foundation for support.

Parts of this work were completed when RK was visiting, 
the IHES in Bures--sur--Yvette, the Max--Planck--Institute in Bonn and the University of Hamburg with a Humboldt fellowship. He gratefully acknowledges
their support. BK acknowledges the hospitality of the DESY theory group where finishing touches for this article were made. 
Both RK and BK thank the Institute for Advanced Study where this version was written

We also wish to thank Sergey Fomin for a short but very valuable discussion.

\end{document}